\documentclass[11pt]{article}

\usepackage{amsfonts}
\usepackage{amsmath}
\usepackage{amssymb}
\usepackage{amsthm}

\usepackage{graphicx}

\theoremstyle{plain}
\newtheorem{theorem}{Theorem}
\newtheorem{proposition}[theorem]{Proposition}

\newtheorem{corollary}[theorem]{Corollary}

\theoremstyle{definition}
\newtheorem{example}[theorem]{Example}
\newtheorem{remark}[theorem]{Remark}

\numberwithin{equation}{section}
\numberwithin{theorem}{section}

\DeclareMathOperator{\diag}{diag}
\newcommand{\D}{\,\mathrm{d}}

\newcommand{\bs}[1]{{\boldsymbol{#1}}}

\newcommand{\R}{\mathbf{R}}

\newcommand{\Ham}[1]{H_{i}}

\begin{document}
\title{{On the behavior of the leading eigenvalue of Eigen's evolutionary matrices}}

\author{Yuri S. Semenov$^{1}$, Alexander S. Bratus$^{1,2}$, Artem S. Novozhilov$^{{3},}$\footnote{Corresponding author: artem.novozhilov@ndsu.edu} \\[3mm]
\textit{\normalsize $^\textrm{\emph{1}}$Applied Mathematics--1, Moscow State University of Railway Engineering,}\\[-1mm]\textit{\normalsize Moscow 127994, Russia}\\[2mm]
\textit{\normalsize $^\textrm{\emph{2}}$Faculty of Computational Mathematics and Cybernetics,}\\[-1mm]
\textit{\normalsize Lomonosov Moscow State University, Moscow 119992, Russia}\\[2mm]
\textit{\normalsize $^\textrm{\emph{3}}$Department of Mathematics, North Dakota State University, Fargo, ND 58108, USA}}

\date{}

\maketitle

\begin{abstract}
We study general properties of the leading eigenvalue $\overline{w}(q)$ of Eigen's evolutionary matrices depending on the probability $q$ of faithful reproduction. This is a linear algebra problem that has various applications in theoretical biology, including such diverse fields as the origin of life, evolution of cancer progression, and virus evolution. We present the exact expressions for $\overline{w}(q),\overline{w}'(q),\overline{w}''(q)$ for $q=0,0.5,1$ and prove that the absolute minimum of $\overline{w}(q)$, which always exists, belongs to the interval $[0,0.5]$. For the specific case of a single peaked landscape we also find lower and upper bounds on $\overline{w}(q)$, which are used to estimate the critical mutation rate, after which the distribution of the types of individuals in the population becomes almost uniform. This estimate is used as a starting point to conjecture another estimate, valid for any fitness landscape, and which is checked by numerical calculations. The last estimate stresses the fact that the inverse dependence of the critical mutation rate on the sequence length is not a generally valid fact. Therefore, the discussions of the error threshold applied to biological systems must take this fact into account.

\end{abstract}

%
%
%
\section{Introduction}
Eigen's quasispecies mathematical model was formulated in 1971 by Manfred Eigen \cite{eigen1971sma} and further expanded by Eigen, his co-authors, and many others (a comprehensive review of what was done by 1990 can be found in \cite{eigen1989mcc}, a review of the main developments in the quasispecies theory in the 90es of the 20th century is \cite{baake1999}, and more recent advances are discussed at length in \cite{jainkrug2007}). This model was originally put forward to facilitate the discussion on various issues concerning the origin of life, and modeled evolution of a population of early replicators (which, e.g., chemically can be thought of as proto RNA-type molecules) subject to the evolutionary forces of selection, which is a consequence of differential reproduction rates, and mutation, which is a consequence of unfair reproduction. Since the early replicators are expected to produce a lot of mutants, therefore it is important to emphasize that the suggested framework can be used for high mutation rates. This last observation made the mathematical framework of the Eigen model one of the main modeling approaches to virus evolution (e.g., \cite{Domingo2012,Lauring2010}). Finally, somewhat late it was realized (see, e.g., \cite{baake1999} for a discussion) that the Eigen model is essentially equivalent to a classical selection--mutation model in a multi-allele haploid population, which usually served as a starting point for the founding fathers of the mathematical genetics, Fisher, Haldane, and Wright, to tackle more realistic, and consequently more mathematically involved, problems (the mathematical theory of the mutation--selection balance is dealt with in, e.g., \cite{burger2000mathematical}). All in all, Eigen's mathematical model must be considered today as one of the main theoretical points of interaction between biology and mathematics, along with other classical examples such as the Lotka--Volterra predator--prey model, the FitzHugh--Nagumo model of spike generation in an axon, or the Fisher--Kolmogorov--Petrovskii--Piskunov model for the spread of an advantageous gene, to mention a few. It is quite difficult to access today its relevance to biology from a quantitative point of view (see, e.g., \cite{Holmes2010}), but it is also impossible to deny its influence on biological theories related to such diverse fields, e.g., as origin of life, virus evolution, evolution of cancer, HIV infection, or evolution of altruism. The influence of the model primarily appears through the use of the observed in the mathematical model phenomena as the metaphors for in vitro or in vivo observed or predicted properties of various living systems.

Two such main metaphors that had risen by the Eigen model are the notions of the \textit{quasispecies} and the \textit{error threshold}. We stress at this point that both of them are somewhat poetically named phenomena that were observed and studied initially within the realm of a mathematical model, and not vice versa. The quasispecies, according to Eigen and co-authors, is a cloud of mutants around the fittest genotype (mathematically speaking, a probability distribution of different observed genotypes in the population). It is this very cloud, how can be elementary proved, is the object of selection that competes with other admissible under the given mutation scheme probability distributions. On the biological point of view there are different opinions whether the quasispecies concept is relevant to, e.g., RNA viruses (see, e.g., \cite{Domingo2012,Holmes2010}), this will not concern us here as far as the quasispecies concept is of significant general interest to many evolutionary processes and many evolutionary biologists.

Mathematically, though, the quasispecies is simply a specific eigenvector of an evolutionary matrix that unites selection and mutation processes and that corresponds to the maximal real positive eigenvalue, which turns out to be equal to the mean population fitness. The \textit{existence} of the quasispecies eigenvector is an elementary consequence of the Perron--Frobenius theorem for nonnegative matrices (e.g., \cite{laub2005matrix}). The Perron--Frobenius theorem, however, does not say anything about the \textit{exact form} of this eigenvector, which is biologically of significant interest; for instance, of particular interest is how much this distribution localized around the fittest genotype. It turns out (and we discuss it more below) that these are exceptional cases when the quasispecies distribution can be written down explicitly for general selection and mutation schemes. Therefore, it is a usual practice either to make some simplifying assumptions about the selection and mutation processes, or resort to numerical computations, which are possible only for very short molecules. The only general case of the Eigen model when, at least in principle, it is feasible to present an explicit solution for the quasispecies distribution corresponds to a multiplicative fitness landscape (i.e., when the fitness contributions of different nucleotides in a molecule are independent). We outline this solution (as first obtained in \cite{rumschitzki1987spectral}) in Appendix \ref{append:A}. We reiterate here that the term ``quasispecies'' in theoretical evolutionary works very often refers to a mathematical object, about which we know for sure that it exists, but in most cases we do not know the exact form of this object.

The second metaphor, the error threshold, that was introduces by Eigen, is more difficult to discuss. First of all, there exists no universally accepted mathematical definition of the error threshold. The verbal descriptions in biologically oriented works can pass the message of the error threshold meaning, but of no help if the question is to provide a mathematically unambiguous definition, which can be used in numerical or analytical computations.
Another important side of the story is that the notion of the error threshold, how it is generally accepted, was formed by numerical computations and quite stunning figures in \cite{swetina1982self}, where no definition is required to actually see the error threshold in figures (now this Fig. 10 in \cite{swetina1982self} became in a sense classical, one example is \cite{szathmary1997rrf} and another representative example is \cite{Domingo2012}, where similar figures are presented as \textit{the} behavior of the solutions to the Eigen model). Together with the impressive figures an estimate of the critical mutation rate was present and compared with the numerical computations, showing a very good accuracy. What is usually missing in the discussions of this work in biologically oriented texts that both the figure and the estimate are specific to the situation when there is one fittest genotype, and everyone else in the population has a lower but the same fitness, this is so-called \textit{single or sharply peaked landscape}.

Our main conclusion from the last three paragraphs is as follows. The classical status of the Eigen quasispecies mathematical model notwithstanding, there is a very serious distinction from other classical mathematical models, e.g., those mentioned above. Whereas we possess an exhaustive mathematical understanding of many mathematical models that are used, among other things, as metaphors in biology, our mathematical understanding of the Eigen model is still somewhat limited. First of all this fact has to be clearly understood when the language and the results of this model are applied for a particular biological system. And second this prompts for more careful and better understanding of mathematical peculiarities of the Eigen model. This second goal is exactly what we try to (partially) achieve in the present manuscript.

There are a great deal of insightful and deep results concerning the mathematical properties of the Eigen model. Several approaches used to derive these results included initially some perturbation technics and numerical analysis \cite{eigen1971sma,eigen1989mcc,eigen1988mqs,schuster1988stationary,swetina1982self}. A special structure of the Eigen evolution matrix and its relevance to the solution of the mathematical model were also early noted (e.g., \cite{dress1988evolution,rumschitzki1987spectral,thompson1974eigen}), and it was also noted that explicit solutions are possible only in very special cases. In 1986 paper \cite{leuthausser1986exact} was published, in which an exact  correspondence between the Ising model of statistical physics and the Eigen evolutionary model was established. This correspondence actually points to the fact that our limited knowledge of the Eigen quasispecies model is due to its intrinsic complexity, because for much longer studied Ising model also there exists no general exact solution. As a results a significant number of papers were published that employed well developed methods of statistical physics to the problem of biological evolution (e.g., \cite{baake1997ising,baake2001mutation,galluccio1997exact,leuthausser1987statistical,saakian2006ese} to mention just a few), we especially would like to single out work by Baake and Wagner \cite{baake2001mutation}, in which not only the results are presented, but also many details are explained for the reader not familiar with statistical physics jargon. In \cite{Baake2007,Hermisson2002} a maximum principle is formulated for a wide class of models, that include the classical Eigen's quasispecies model. While this maximum principle provides a very powerful tool for analysis of the Eigen quasispecies model, it still requires several simplifying assumptions, in particular, permutation invariant structure of the fitness landscape, and smooth approximation of the fitness functions involved in the limit of infinite genome length. Our goal in this text is to present exact results for finite dimensional problems, without resorting to any limit procedures and scalings.

Our approach is based on analysis of the spectral properties of the Eigen evolutionary matrix (similar to what was done in \cite{garcia2002linear,rumschitzki1987spectral}, and what we applied for the analysis of another incarnation of the Eigen model --- permutation invariant Crow--Kimura model in \cite{bratus2013linear}). We attempt to use the special structure of the matrices to obtain an insight about qualitative behavior of several important quantities, the central of which is the mean population fitness or, mathematically, the leading eigenvalue. The results that form the core of the present manuscript include analytical analysis of the first and second derivative of the mean population fitness $\overline{w}(q)$ as a function of the mutation probability $1-q$. In particular, we prove that this function always possesses its absolute minimum for $0<q<0.5$. The exact expressions for $\overline{w}(q),\overline{w}'(q),\overline{w}''(q)$ are presented for $q=0,0.5,1$. The single peaked landscape is used as an example to apply the general analytical results for an arbitrary fitness landscape, and upper and lower bounds for $\overline{w}(q)$ are given for any $0\leq q\leq 1$. We show that the estimate of the critical mutation probability as given by Eigen and many others is always an upper estimate for an arbitrary fitness landscape, which can be very inaccurate in general. These exact mathematical results can be used to obtain a better intuitive understanding of the behavior of the mean population fitness for a general Eigen evolution matrix. We also compare a heuristic estimate, obtained under condition of the uniform distribution of the quasispecies, of the critical mutation probability with numerical computations.

The rest of the text is organized as follows. In Section \ref{sec:2} we introduce the notations and state the mathematical problem. There are a lot of mathematical models in the literature that go under the umbrella of the Eigen model. These models may include, in addition to the selection and mutation processes, recombination, non-constant sequence length, duplications and deletions, horizontal gene transfer, and other possible generalizations. We do not consider those. In Section \ref{sec:2} we formulate a simplest Eigen model, which still possess the salient features important for biological applications. In Section \ref{sec:3} we state the main results, illustrate them, and discuss possible biological interpretations; no proofs are given here to streamline the exposition. Sections \ref{sec:4} and \ref{sec:5} are devoted to a detailed presentation of the proofs of the results discussed in Section \ref{sec:3}. Finally some additional mathematical results are presented in Appendix.

\section{Problem statement and notations}\label{sec:2}
Consider an infinite population of sequences of fixed \textit{length} $N$. Each sequence is composed of the letters from the alphabet $\{0,1\}$, therefore the total number of different \textit{types} of sequences is $l:=2^N$. There are different possible ways to index all the sequence types in the population, we choose the one when sequence type $i$ is exactly the binary representation of the number $i$:
$$
\text{Sequence type }i=[a_0,a_1,\ldots,a_{N-1}]\Longleftrightarrow i=a_0+a_1 2+\ldots+ a_{N-1}2^{N-1},\,a_i\in\{0,1\},
$$
note the order of letters $a_i$ in the sequence. For example, sequence $\mbox{type }0=[0,0,\ldots,0]$, and sequence type $l-1=2^N-1$ is the sequence of all 1s: $\mbox{type }l-1=[1,1,\ldots,1]$. Using this lexicographical order of sequence types, we use the index variables that run usually from 0 to $l-1=2^N-1$ or from $0$ to $N-1$.

Let $p_i(t)$ denote the frequency of the sequences of type $i$, and $\bs p(t)=\bigl(p_0(t),\ldots,p_{l-1}(t)\bigr)^\top$, where $^\top$ means transposition and bold letters denote matrices and vectors, which are always assumed to be column-vectors. There are two basic evolutionary processes in the population, namely, the \textit{selection process}, which is modeled by considering differential reproduction rates, or \textit{fitnesses}, $w_i>0,\,i=0,\ldots l-1$, of the sequence types, which we put together as the diagonal matrix $\bs W=\diag(w_0,\ldots,w_{l-1})\in \R^{l\times l}$ or as the vector $\bs w=(w_0,\ldots,w_{l-1})^\top\in\R^l$; both $\bs W$ and $\bs w$ are called the \textit{fitness landscape}. The second evolutionary process in the population is the \textit{mutation process}, which is described by the probabilities of producing offspring of type $i$ upon a reproduction event of a sequence of type $j$: $q_{ij}=\mbox{Pr}\,(\text{sequence }j\to \text{sequence }i)$. The mutation matrix hence is $\bs Q=(q_{ij})_{l\times l}$, in the following we also use the notation $\bs Q_N$ for the matrix of dimension $2^N\times 2^N$. Putting together the assumptions on the selection and mutation processes, which are coupled in the original Eigen model, we obtain a system of ordinary differential equations
(see, e.g., \cite{eigen1971sma,eigen1989mcc,eigen1988mqs} for the original work, or \cite{baake1999,jainkrug2007} for recent reviews)
\begin{equation}\label{eq0:1}
    \bs{\dot p}(t)=\bs{QWp}(t)-\overline{w}(t)\bs p(t),
\end{equation}
where the term
$$
\overline{w}(t)=\sum_{i=0}^{l-1} w_ip_i(t)=\bs w\cdot \bs p(t),
$$
which describes the \textit{mean population fitness}, is necessary to keep the condition $\sum_{i=0}^{l-1}p_i(t)=1$ for any $t$. Everywhere in the text the dot denotes the standard inner product in $\R^l$: $\bs x\cdot \bs y=\sum_{i=0}^{l-1}x_iy_i$ for any two vectors $\bs x,\bs y\in \R^l$.

To further define the main object of our analysis, we can assume that the mutation probabilities are independent of the position in the sequence and other possible mutations. Therefore, if $q\in[0,1]$ denotes the probability of \textit{error-free} reproduction per one site of a sequence per replication event, then
$$
q_{ij}=q^{N-H_{ij}}(1-q)^{H_{ij}},
$$
where $H_{ij}$ is the \textit{Hamming distance} between sequences $i$ and $j$, which is defined as the number of sequence sites at which sequences $i$ and $j$ are different.

Model \eqref{eq0:1} is essentially linear, and its behavior is determined by the leading eigenvalue and the corresponding eigenvector of matrix $\bs{QW}$, which we call the \textit{Eigen evolutionary matrix}. To be precise, if $\bs{QW}$ is \textit{irreducible} (recall that $\bs A=(a_{ij})_{n\times n}$ is irreducible if an associated with $\bs A$ directed graph with $n$ vertices, for which there is an edge from vertex $j$ to vertex $i$ if and only if $a_{ij}>0$, is strongly connected, i.e., from any vertex $i$ there is a path to any vertex $j$), then there is a unique globally stable equilibrium of the dynamical system \eqref{eq0:1}:
$$
\lim_{t\to\infty} \bs p(t)=\bs p,
$$
such that $\sum_{i=0}^{l-1}p_i=1$ and $\bs p>0$ for any positive initial conditions.

This equilibrium $\bs p$ was called by Manfred Eigen and his co-authors the \textit{quasispecies}. Quasispecies $\bs p$ can be found as the right eigenvector of matrix $\bs{QW}$ corresponding to the leading (Perron--Frobenius) eigenvalue, which is given by $\overline{w}=\bs{Wp}=\bs w\cdot \bs p>0$. We are interested in the exact form and characteristics of the distribution $\bs p$ and of the mean fitness $\overline{w}=\bs w\cdot \bs p$ for given $\bs w$ and $q$. It turns out that the special form of the mutation matrix $\bs Q$ allows us to obtain a number of analytical results using the standard methods of linear algebra.

\section{Main results and discussion}\label{sec:3}In this section we present the main results of this manuscripts and illustrate them with examples. The proofs and various additional details are delegated to the rest of the text.

\subsection{Analytical results for a general fitness landscape $\bs w$}
We consider the eigenvalue problem
\begin{equation}\label{main:1}
    \bs{QWp}=\overline{w} \bs p,
\end{equation}
where $\bs p=(p_0,\ldots,p_{l-1})^\top$ is the positive eigenvector corresponding to the leading eigenvalue $\overline{w}$ with the normalization
\begin{equation}\label{main:2}
    \sum_{i=0}^{l-1}p_i=1,\quad \overline{w}=\bs{w}\cdot \bs p.
\end{equation}
Both the leading eigenvalue $\overline{w}$ and the corresponding eigenvector $\bs p$ are functions of the mutation parameter $q$ (the probability of the honest reproduction per site per replication event):
$$
\overline{w}=\overline{w}(q),\quad \bs p=\bs p(q),
$$
and our main goal is to study the behavior of $\overline{w}(q)$ depending on $q\in[0,1]$.

It turns out that it is possible to exactly calculate $\overline{w}(q)$ and its first derivative with respect to parameter $q$ for some specific values of $q$. In particular, if the length $N$ of sequences in the population satisfies  $N\geq 2$ and the eigenvalues $\overline{w}(0)$ and $\overline{w}(1)$ of matrix $\bs{QW}$ have multiplicity one then we have the following Table~\ref{tab:main:1}.
\begin{table}[!ht]
  \centering
  \begin{tabular}{ |c | c | c | c|}
  \hline&&&\\[-3mm]
  $q$&$0$&$0.5$&$1$\\[1mm]
  \hline&&&\\[-3mm]
  $\overline{w}(q)$&$\max \{\sqrt{w_iw_{i^\ast}}\}$&$\displaystyle\frac{1}{2^N}\sum_{i=0}^{2^N-1}w_i$&$\max\{w_i\}$\\[5mm]
  \hline&&&\\[-3mm]
  $\overline{w}'(q)$ & $-N\max\{\sqrt{w_i w_{i^\ast}}\}$ & $\displaystyle \frac{1}{\overline{w}(0.5)2^{2N-1}}\,\bs w\cdot \bs A\bs w$ & $N\max\{w_i\}$\\[3mm]
  \hline
\end{tabular}
  \caption{Expressions for $\overline{w}(q)$ for particular $q$. Here $i^\ast$ is the index conjugate to $i$, formally, $i^\ast=2^N-1-i$. Matrix $\bs A:=N\bs E-2\bs H$, where $\bs E$ is the matrix of all ones, and $\bs H$ is the matrix with elements $H_{ij}$, where $H_{ij}$ is the Hamming distance between sequences $i$ and $j$ }\label{tab:main:1}
\end{table}

Moreover, for $N\geq 3$ it is possible to calculate $\overline{w}''(q)$ for $q=0,0.5,1$, see Table \ref{tab:N:1}.

The methods we used to calculate the expressions in Table \ref{tab:main:1} together with the analysis of the second derivative of $\overline{w}(q)$ lead to the following theorem.
\begin{theorem}\label{theor:main:1}Consider eigenvalue problem \eqref{main:1} together with normalization \eqref{main:2}. If the fitness landscape is such that the leading eigenvalues $\overline{w}(0)$ and $\overline{w}(1)$ have multiplicity 1, then there exists an absolute minimum $\hat{w}$ of function $\overline{w}(q)$ for $0<q\leq 0.5$. The point of this minimum $\hat{w}=\overline{w}(\hat{q})$ is determined by the condition $\overline{w}'(\hat{q})=0$. For $q\geq 0.5$ function $\overline{w}(q)$ is non-decreasing and convex.
\end{theorem}
A typical example of $\overline{w}(q)$ is given in Figure~\ref{fig:main:1}.
\begin{figure}[!t]
\centering
\includegraphics[width=0.5\textwidth]{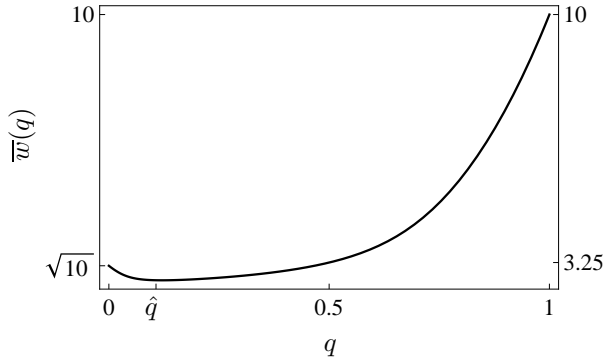}
\caption{A typical picture for a qualitative behavior of $\overline{w}(q)$ depending on $q$. In this particular example the sequence length is $N=3$, the fitness landscape if $\bs w=(10,3,3,2,3,2,2,1)$, therefore (see Table \ref{tab:main:1}), $\overline{w}(1)=10$, $\overline{w}(0)=\sqrt{10}$, $\overline{w}(0.5)=\frac{13}{4}=3.25$}\label{fig:main:1}
\end{figure}

From a biological standpoint we are mostly interested in the behavior of $\overline{w}$ for $q$ close to 1 (it is difficult to imagine a prospering population with mutation probabilities $\geq 0.5$), and for this range of parameter values the technical condition that the leading eigenvalue $\overline{w}(1)$ has multiplicity one simply means that there is a unique maximum in the fitness landscape $\bs w$.

\subsection{Results for the single peaked landscape}\label{sec:3:2}
The general results from the previous subsection can be further used to study particular fitness landscapes. One of the most famous (because of its simplicity and amenability to calculations, not because of its biological relevance) is the so called \textit{single} or \textit{sharply peaked landscape}, when there is a unique master sequence with the maximal fitness, and all other sequences have equal fitnesses lower than that of the master one. Single peaked landscape is the extreme form of \textit{positive} (or \textit{diminishing return}) \textit{epistasis}, when only the mutations in the master sequence cause (negative) changes in the fitness of the mutants. Subsequent mutations (if they do not imply the master sequence again) do not lead to further changes for the offspring fitness. If, however, the mutations have an aggravating effect, this is termed \textit{negative} (or \textit{synergetic}) \textit{epistasis} \cite{jainkrug2007}.

The single peaked landscape is defined as
$$
\bs w^\top=(w,\ldots,w,w+s,w,\ldots,w),
$$
where $w$ is a fixed positive constant and $s>0$ is the selective advantage of the master sequence, which has index $k$, $0\leq k\leq l-1$. In the literature the usual positioning of the master sequence is $k=0$, so that the fittest sequence has the composition of all zeros. Table~\ref{tab:main:1} for the single peaked landscape turns into Table \ref{tab:main:2}, where also the expressions for the second derivative are given.
\begin{table}[!hb]
  \centering
  \begin{tabular}{| c | c |c| c|}
  \hline&&&\\[-3mm]
  $q$&$0$&$0.5$&$1$\\[1mm]
  \hline&&&\\[-3mm]
  $\overline{w}(q)$&$\sqrt{w(w+s)}$&$\displaystyle w+\frac{s}{2^N}$& $w+s$\\[3mm]
  \hline&&&\\[-3mm]
  $\overline{w}'(q)$&$-N\sqrt{w(w+s)}$&$\displaystyle \frac{Ns^2}{(2^Nw+s)2^{N-1}}$&$N(w+s)$\\[4mm]
  \hline&&&\\[-3mm]
  $\overline{w}''(q)$&{\footnotesize$\displaystyle N\sqrt{w(w+s)}\left(N+\frac{4w}{s}\right)$} &{\footnotesize$\displaystyle \frac{Ns^2\bigl((N-1)s^2+2^{N+2}Nws+2^{2N}(N+1)w^2\bigr)}{2^{N-2}(s+w2^{N})^3}$}&{\footnotesize $\displaystyle N(w+s)\left(N-1+\frac{2w}{s}\right)$}\\[3mm]
  \hline
\end{tabular}
  \caption{Expressions for $\overline{w}(q)$, $\overline{w}'(q)$, and $\overline{w}''(q)$ for the single peaked landscape}\label{tab:main:2}
\end{table}

As a simple corollary of Theorem \ref{theor:main:1} we find that the minimum value of $\overline{w}(q)$ is found for $0<q\leq 0.5$, and for $0.5\leq q\leq 1$ function $\overline{w}(q)$ is nondecreasing and convex. This is an exact theoretical result valid for any finite sequence length $N$. However, as everyone dealing with Eigen's quasispecies model knows, this particular fitness landscape is a source of very impressive figures that illustrate the phenomenon dubbed by Eigen as the \textit{error threshold} \cite{eigen1988mqs}.

The error threshold is difficult to rigourously define (see \cite{Hermisson2002} for a detailed discussion of various threshold-like behaviors in the quasispecies model), but for the sake of the present discussion it suffices to approximate it in the following way. Assume that the fittest type has the index zero, and also assume that there are no mutations from any sequences to sequence 0. Then the first equation in \eqref{eq0:1} would read
$$
\dot p_0=q^Nw_0p_0-\overline{w}p_0.
$$
Assuming that $q^Nw_0=\overline{w}$ and $p_0\to 0$, we find that
$$
q^N=\frac{\sigma}{w_0}\,,
$$
where
$$
\sigma=\sum_{i=1}^{l-1}w_ip_i.
$$
Assuming the single peaked fitness landscape, i.e., $w_0=w+s$, $w_i=w$, we find
\begin{equation}\label{main:3a}
q^\ast_{\text{Eigen}}=\sqrt[N]{\frac{w}{w+s}}\,.
\end{equation}
or using logs and approximations,
\begin{equation}\label{main:3}
q_{\text{Eigen}}=1-\frac{\log \frac{w+s}{w}}{N}\,,
\end{equation}
    This expression becomes exact when $N\to \infty$.

Of course this expression for any meaningful values of $w,s,N$ is quite close to 1. The error threshold signifies the transition from a localizes quasispecies, when most of the sequences in the population are the master sequence and its close mutants, to the case when the distribution of the sequences is close to uniform. To unite these observations with the exact statement of Theorem \ref{theor:main:1}, we have to conjecture that the graph of $\overline{w}(q)$ should possess a very flat plateau, which starts approximately at the value $q_{\text{Eigen}}$ on the right to some value on the left, and the exact minimum $\hat{w}$ should be quite close to $\overline{w}(q_{\text{Eigen}})$. In general, the critical mutation rate the closer to 1 the longer the sequence. This conjecture can be supported by numerical computations (see figures below). Therefore, to present mathematically rigorous results that would be of some value for the biological interpretation of the model, we need to estimate the point at which the graph of function $\overline{w}(q)$ becomes \textit{almost} flat for $q$ close to 1. Our estimate is based on the following proposition:
\begin{proposition}\label{prop:main:1}Consider the Eigen quasispecies model with the single peaked landscape $\bs w^\top=(w+s,w,\ldots,w)$ for $w,s>0$. If $0\leq q\leq 0.5$ then the mean fitness $\overline{w}(q)$ satisfies
\begin{equation}\label{main:4}
    w+sq^N\leq \overline{w}(q)\leq \sqrt{w^2+sw(1-q)^N+s^22^{-2(N+1)}}+s2^{-N-1}.
\end{equation}
If $0.5\leq q\leq 1$ then the mean fitness $\overline{w}(q)$ has a lower bound
\begin{equation}\label{main:5}
\max\left\{w+s2^{-N},(w+s)q^N\right\} \leq   \overline{w}(q)
\end{equation}
and an upper bound
\begin{equation}\label{main:6}
\overline{w}(q)\leq \frac{w+sq^N}{2}+\sqrt{\left(\frac{w+sq^N}{2}\right)^2-ws\bigl(q^N-(2q^2-2q+1)^N\bigr)}\,.
\end{equation}
\end{proposition}

An illustration of Proposition~\ref{prop:main:1} is given in Figure \ref{fig:main:2}, where it can be seen that the estimates \eqref{main:5} and \eqref{main:6} are quite tight.
\begin{figure}[!b]
\centering
\includegraphics[width=0.5\textwidth]{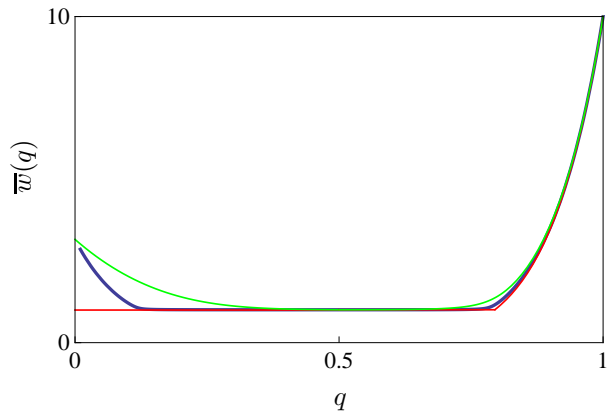}
\caption{Illustration of Proposition~\ref{prop:main:1}. The green curve shows the upper bounds in \eqref{main:4} and \eqref{main:6}, the red curve shows the lower bounds in \eqref{main:4} and \eqref{main:5}, and the blue curve is the numerically computed $\overline{w}(q)$ for the single peaked fitness landscape $\bs w=(10,1,\ldots,1)$ and sequence length $N=10$}\label{fig:main:2}
\end{figure}

Moreover, we note that in the lower estimate \eqref{main:5}, the second of the two functions is used first, and after some $q^\ast$ another one is used. The point of intersection $q^\ast$ can be found as
\begin{equation}\label{main:7}
    q^\ast=\sqrt[N]{1-\frac{s(1-2^{-N})}{w+s}}=\sqrt[N]{\frac{w+s2^{-N}}{w+s}}=\sqrt[N]{\frac{\overline{w}(0.5)}{\overline{w}(1)}}\,.
\end{equation}
This is almost exactly expression \eqref{main:3a} with a very small difference in $s2^{-N}$ term. In Figure \ref{fig:main:3} four examples are shown how the estimate $q^\ast$ works for different sequence lengths.
\begin{figure}[!b]
\centering
\includegraphics[width=0.9\textwidth]{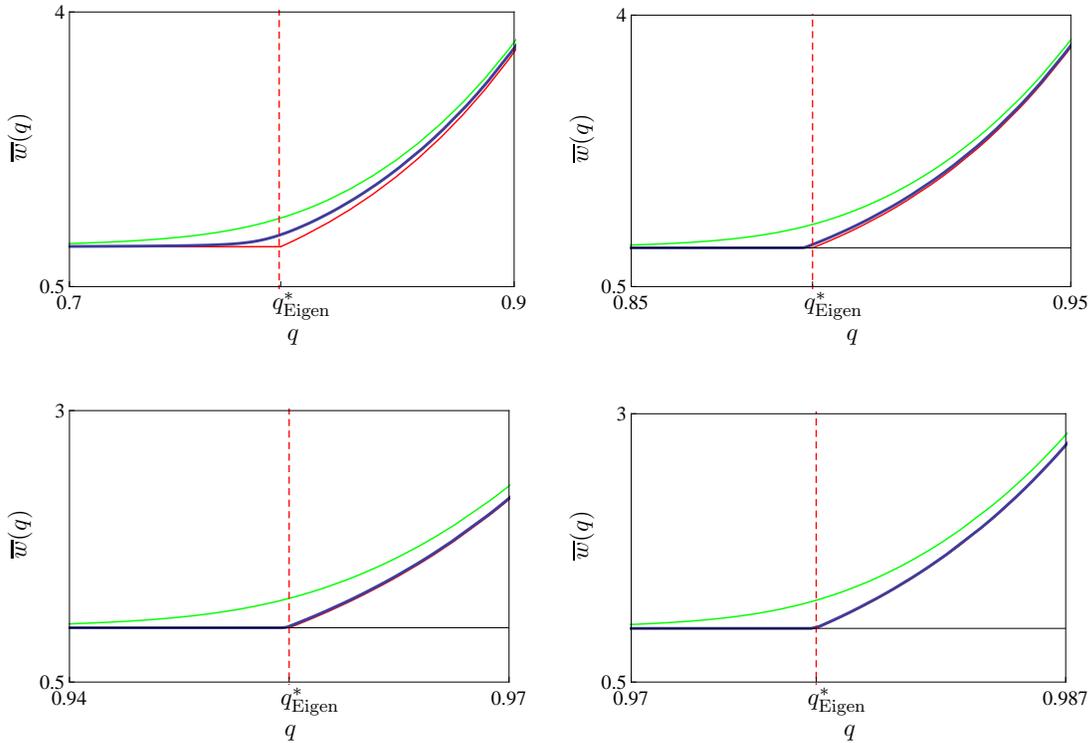}
\caption{Estimate \eqref{main:7} of the error threshold (dotted red line) for the single peaked fitness landscape $\bs w=(10,1,\ldots,1)$ for sequence lengths $N=10,20,50,100$ (upper left, upper right, lower left, lower right respectively). The green curves are the upper bound \eqref{main:6} and the red curves are the lower bound \eqref{main:5}. The blue curves are the graphs of $\overline{w}(q)$ calculated numerically. Note a sudden change in the slope of $\overline{w}(q)$, especially for larger $N$}\label{fig:main:3}
\end{figure}

\subsection{Relationship of the single peaked landscape and an arbitrary fitness landscape}\label{sec:3:3}The reason that we devoted so much space to the single peaked landscape is not only its analytical tractability, but also the following almost obvious fact.
\begin{proposition}\label{prop:3:3}Consider two Eigen's models with fitness landscapes $\bs w^{(1)}=(w_0^{(1)},\ldots,w_{l-1}^{(1)})$ and $\bs w^{(2)}=(w_0^{(2)},\ldots,w_{l-1}^{(2)})$. If $\bs w^{(1)}\geq \bs w^{(2)}$ then
$$
\overline{w}^{(1)}(q)\geq \overline{w}^{(2)}(q),\quad 0\leq q\leq 1.
$$
\end{proposition}
In particular, if the fitness landscape is $\bs w=(w_0,\ldots,w_{l-1})$, and $w+s=\max_ i\{w_i\}, w=\max_{i}\{w_i\colon w_i\neq w+s\}$, then the single peaked landscape $\bs w_{\text{SPL}}=(w,\ldots,w,w+s,w,\ldots,w)$ dominates $\bs w$. In a similar vein a lower boundary by a corresponding single peaked landscape can be constructed.

Therefore, we obtain an important corollary.
\begin{corollary}\label{corr:3:4}Consider a fitness landscape $\bs w=(w_0,\ldots,w_{l-1})$ for the Eigen model, such that the maximal fitness is $w_0$. Denote $w+s=w_0$ and $w=\max_{i\geq 1}\{w_i\}$. Then, if there exists the error threshold for the model with $\bs w$, then its critical mutation rate $q^\ast$ if bounded above by $q^\ast_{\text{\emph{SFL}}}$, calculated by \eqref{main:7} with $w+s$ and $w$ as defined above.
\end{corollary}

To illustrate the last corollary, consider the fitness landscape defines as
\begin{equation}\label{main:8}
    w_i=s^{(H_i)^\alpha},\quad i=0,\ldots,l-1,
\end{equation}
where $0<s<1$, $H_i$ as before the Hamming norm of the index $i$, and $\alpha>0$ is the parameter that defines the epistasis in the system. If $\alpha<1$ then we have positive epistasis, and for $\alpha>1$ we have negative epistasis. If $\alpha=1$ then the fitness landscape is multiplicative, and it is possible to write down the explicit solution for $\bs p(q)$ and, consequently, for $\overline{w}(q)$ (see Appendix \ref{append:A}). We are especially interested in values $\alpha<1$.

The results of the computations are given in Figure \ref{fig:main:4}.
\begin{figure}[!t]
\centering
\includegraphics[width=0.9\textwidth]{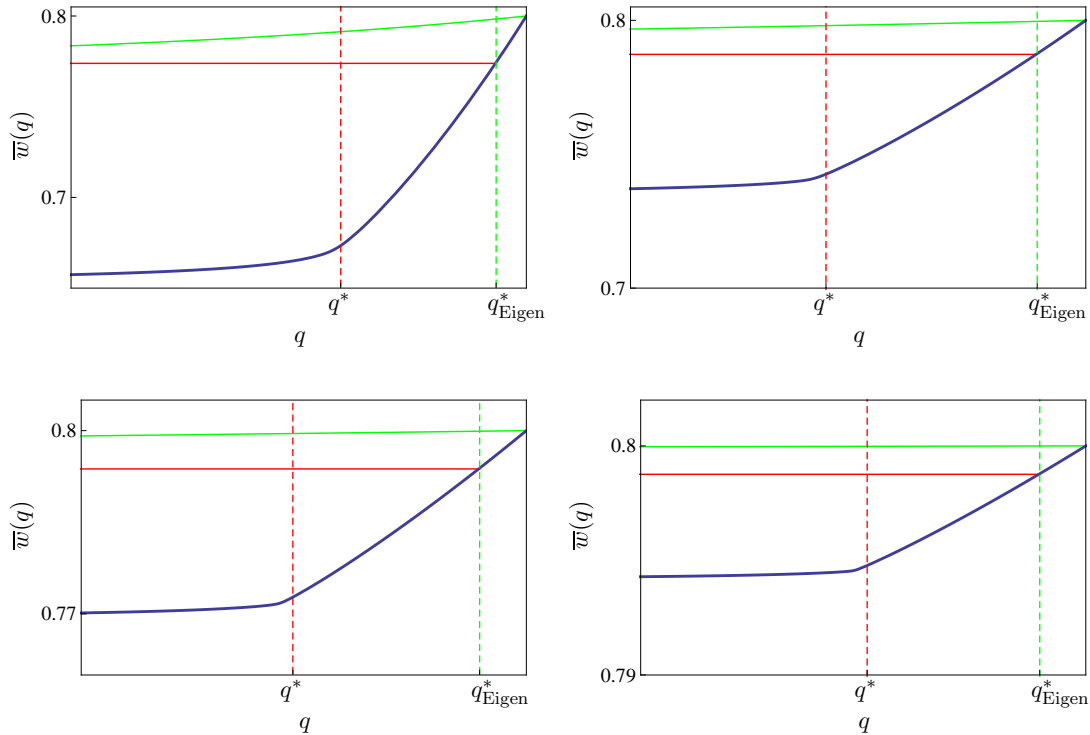}
\caption{Illustration of Corollary \ref{corr:3:4}. The Eigen model with fitness landscape \eqref{main:8} is considered. The cases of $s=0.8$ and $\alpha=0.2,0.1,0.05,0.01$ are shown. The blue curve is the leading eigenvalue $\overline{w}(q)$. The red dotted line shows the estimate $q^\ast$ by \eqref{main:8a}, and the green dotted line is $q^\ast_\text{Eigen}$, calculated by \eqref{main:3a}. The green and red lines show upper and lower boundaries \eqref{main:5} and \eqref{main:6} for the single peaked landscape $\bs w=(w+s,w,\ldots,w)$, where $w+s=\max_i\{w_i\}$ and $w=\max_{i}\{w_i\colon w_i\neq w+s\}$}\label{fig:main:4}
\end{figure}
Together with $\overline{w}(q)$ (blue), the upper (green) and lower (red) bounds of the corresponding single peaked landscape, which give upper estimate of the critical mutation rate $q^\ast_{\text{Eigen}}$, the expression
 \begin{equation}\label{main:8a}
    q^{\ast}=\sqrt[N]{\frac{\overline{w}(0.5)}{\overline{w}(1)}}=\frac{1}{2}\sqrt[N]{\frac{\sum w_i}{\max\{w_i\}}}\,.
\end{equation}
is calculated. An obvious inspiration for \eqref{main:8a} is the exact result for the single peaked landscape \eqref{main:7}. However, a very similar formula can be obtained if we start again with the equation
$$
q^N=\frac{\sigma}{w_0}\,,\quad \sigma=\sum_{i=1}^{l-1}w_ip_i,
$$
and assume that $w_0=\max\{w_i\}$ and the distribution of $\bs p$ is almost uniform,
$$
p_i\approx \frac{1}{2^N}\,,\quad i=1,\ldots l-1.
$$

A perfunctory inspection shows that the critical mutation rate calculated with the help of \eqref{main:8a} gives a remarkable agreement with observable sudden change of the slope of function $\overline{w}(q)$. Moreover, the estimate \eqref{main:3a} ($q^\ast_\text{Eigen}$) gives a very bad indication of the actual point of the error threshold. We therefore conjecture that formula \eqref{main:8a} can be used for \textit{any} fitness landscape to predict the critical mutation rate, and this estimate will give much better result than the usually employed estimate \eqref{main:3} or \eqref{main:3a}, which significantly overestimates the critical mutation rate.

This estimate \eqref{main:8a} for the critical mutation rate emphasizes the following point: the inverse relationship with the sequence length as in \eqref{main:3} does not pertain to \textit{any} possible fitness landscape, as shown in Figure \ref{fig:main:30}.

It is possible to have fitness landscapes such that the critical mutation rate, according to the formula \eqref{main:8a}, approaches 0.5, which simply means that there are no pronounced sharp transitions in this particular model (recall from Table \ref{tab:main:1} that for $q=1/2$ independently of the fitness landscape we have uniform quasispecies distribution).

\begin{figure}[!t]
\centering
\includegraphics[width=0.49\textwidth]{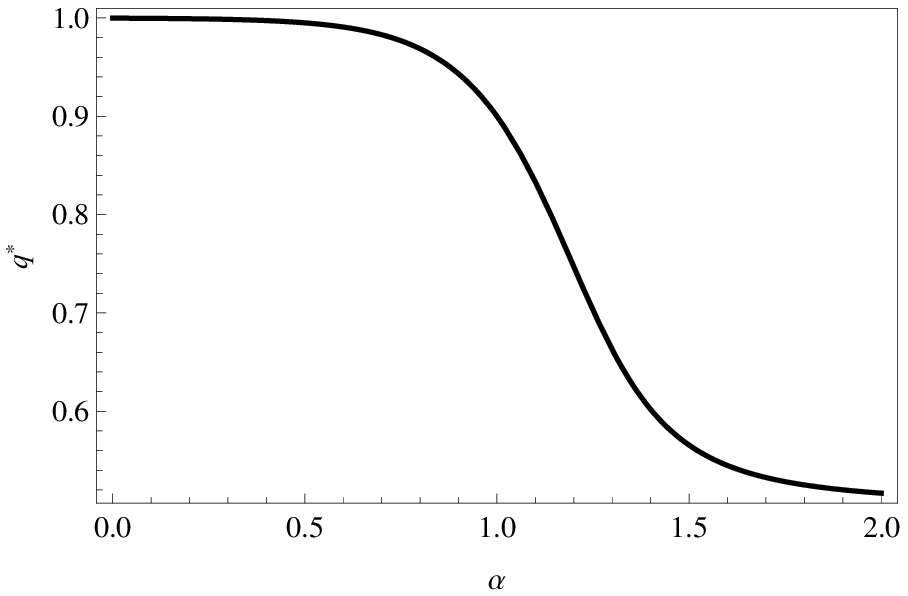}\hfill
\includegraphics[width=0.49\textwidth]{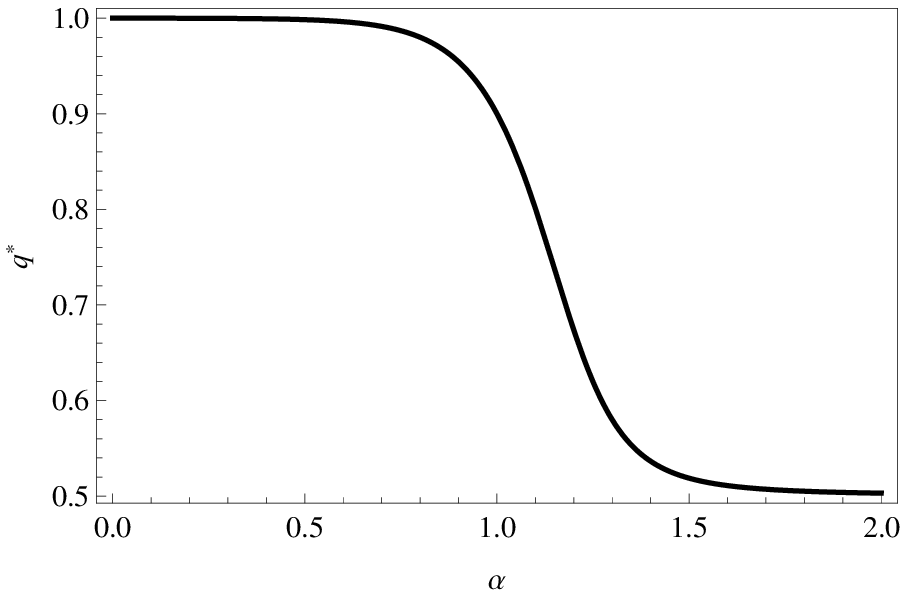}
\caption{Estimate \eqref{main:8a} of the critical mutation rate versus $\alpha$ for the fitness landscape \eqref{main:8}. Left: $N=1000$, Right: $N=10000$}\label{fig:main:30}
\end{figure}

\section{Analytical investigation in the general case}\label{sec:4}
In this section we present proofs of the general and specific results discussed above. In particular, our goal is to give our proof of Theorem \ref{theor:main:1} and the calculations leading to Table \ref{tab:main:1} and Table \ref{tab:main:2}. Some additional details are delegated to Appendix. We start with the properties of the mutation matrix $\bs Q=\bs Q_N$, which will be of paramount importance for the rest of the exposition. Most of the results in the first subsection are well known (see \cite{rumschitzki1987spectral} for the first paper that treats $\bs Q_N$ in a similar way), and we present them to make the manuscript self-contained.
\subsection{Properties of mutation matrix $\bs Q_N$}\label{propQ}
Recall that $\bs Q_N=(q_{ij})_{l\times l},\,l=2^N$, where
\begin{equation}\label{eq1:1}
    q_{ij}=q^{N-H_{ij}}(1-q)^{H_{ij}},
\end{equation}
$q$ is the probability of an error-free copying per site per replication event, and $H_{ij}$ is the Hamming distance between sequences $i$ and $j$.

Below we list the properties of $\bs Q_N$, which are important for the subsequent proofs.
\begin{enumerate}
\item Matrix $\bs Q_N$ is symmetric, which follows from the fact that the Hamming distance is symmetric: $H_{ij}=H_{ji}$. Therefore, there are always $l=2^N$ linearly independent real eigenvectors of $\bs Q_N$, which, consequently, form a basis of $\R^{l}$.
\item The sum of elements of any row or any column of $\bs Q_N$ is equal to one:
$$
\sum_{k=0}^{l-1}q_{kj}=\sum_{k=0}^{l-1}q_{ik}=1.
$$
Therefore, $\bs Q_N$ is \textit{doubly stochastic}.

\item As was noted in \cite{rumschitzki1987spectral} and further scrutinized in \cite{dress1988evolution}, one has a natural recursive process
\begin{equation}\label{eq1:2}
    \bs Q_{N+1}= \bs Q_1\otimes \bs Q_{N},\quad N=1,2\ldots,
\end{equation}
where
$$
\bs Q_1=\begin{bmatrix}
         q & 1-q \\
         1-q & q \\
       \end{bmatrix}.
$$

Here
$$
\bs C=\bs A\otimes \bs B
$$
denotes the \textit{Kronecker product} $\bs C=(c_{ij})_{mp\times nq}$ of matrices $\bs A=(a_{ij})_{m\times n}$ and $\bs B=(b_{ij})_{p\times q}$, defined as
$$
\bs A\otimes\bs B=\begin{bmatrix}
                    a_{11}\bs B & \ldots & a_{1n}\bs B \\
                    \vdots & \ddots & \vdots \\
                    a_{m1}\bs B & \ldots & a_{mn}\bs B \\
                  \end{bmatrix}.
$$

If $\bs A$ and $\bs B$ are square then $\bs A\otimes \bs B$ is also square. The Kronecker product is associative and noncommutative.

The most useful for us property of the Kronecker product is given in the following theorem (see, e.g.,~\cite{laub2005matrix} for more details):
\begin{theorem}\label{theor:4:1}Let $\bs A\in\R^{n\times n}$ have eigenvalues $\lambda_1,\ldots,\lambda_n$, and let $\bs B\in\R^{m\times m}$ have eigenvalues $\mu_1,\ldots,\mu_m$. Then $mn$ eigenvalues of $\bs A\otimes \bs B$ are
$$
\lambda_1\mu_1,\lambda_1\mu_2,\ldots, \lambda_1\mu_n,\lambda_2\mu_1,\ldots,\lambda_m\mu_n.
$$
Moreover, if $\bs x_1,\ldots \bs x_p,\,p\leq m$ are linearly independent right eigenvectors of $\bs A$ corresponding to $\lambda_1,\ldots,\lambda_p$, and $\bs y_1,\ldots,\bs y_q,q\leq n$ are linearly independent right eigenvectors of $\bs B$ corresponding to $\mu_1,\ldots, \mu_q$, then $\bs x_i\otimes\bs y_j\in\R^{mn}$ are linearly independent right eigenvectors of $\bs A\otimes \bs B$ corresponding to $\lambda_i\mu_j$.
\end{theorem}

This theorem together with the fact that matrix $\bs Q_1$ has eigenvalues $1$ and $2q-1$ implies that $\bs Q_N$ has eigenvalues
$1,2q-1,(2q-1)^2,\ldots,(2q-1)^N$ with multiplicities $\binom{N}{0},\binom{N}{1},\binom{N}{2},\ldots,\binom{N}{N}$ respectively. This can be written succinctly as
$$
\lambda_i=(2q-1)^{\Ham{i}},\quad i=0,\ldots, l-1,
$$
where $\Ham{i}:=H_{0i}$ is the \textit{Hamming norm} of (number of 1s in the binary representation of) sequence type $i$.

\item From the previous it immediately follows that
$$
\det \bs Q_N=(2q-1)^{2^{N-1}N}.
$$
In particular, for $2q\neq1$, $\bs Q_N$ is invertible.

\item Two linearly independent eigenvectors of $\bs Q_1$ can be taken as $(1,1)^\top$ and $(1,-1)^\top$. Therefore, by Theorem \ref{theor:4:1}, we find that $\bs Q_N$ has exactly $l=2^N$ linearly independent eigenvectors, and the $i$-th eigenvector of $\bs Q_N$ can be found as the $i$-th column of matrix $\bs T_N$:
$$
\bs T_{N+1}=\bs T_1\otimes \bs T_{N},\quad N=1,2,\ldots,
$$
where
$$
\bs T_1=\begin{bmatrix}
          1 & 1 \\
          1 & -1 \\
        \end{bmatrix}.
$$
In other words,
$$
\bs T_{N}^{-1}\bs Q_N\bs T_N=\diag\bigl(1,2q-1,\ldots,(2q-1)^{H_i},\ldots,(2q-1)^N\bigr).
$$
\end{enumerate}
Now consider some properties of matrix $\bs T_N$:
\begin{enumerate}
\item Matrix $\bs T_N$ is symmetric as the Kronecker product of two symmetric matrices.
\item The inverse $\bs T^{-1}_N$ can be found from the identity
$$
\bs T_N^2=(\bs T_1\otimes\bs T_{N-1})(\bs T_1\otimes\bs T_{N-1})=\bs T_1^2\otimes \bs T_{N-1}^2=\begin{bmatrix}
                                                                                                        2 & 0 \\
                                                                                                        0 & 2 \\
                                                                                                      \end{bmatrix}\otimes \bs T_{N-1}^2.
$$
Therefore, $\bs T_N^2=2^N\bs I,$ and hence
$$
\bs T^{-1}_{N}=2^{-N}\bs T_N.
$$
\item The column (row) with index 0 of $\bs T_N$ is composed of all 1s, all other columns and rows have equal numbers of 1s and $-1$s.
\item The determinant of $\bs T_N$ is
$$
\det \bs T_N=(-2)^{2^{N-1}N}.
$$
\end{enumerate}

\subsection{Analysis of the first derivative of $\overline{w}(q)$}\label{sec:4:2}Here we embark on the study of the basic eigenvalue problem, which we write here again for convenience,
\begin{equation}\label{eq3:1}
    \bs{QWp}=\overline{w} \bs p,
\end{equation}
where $\bs p=(p_0,\ldots,p_{l-1})^\top$ is the eigenvector corresponding to the leading eigenvalue $\overline{w}$ with the normalization
\begin{equation}\label{eq3:2}
    \bs p\in S_l,\quad \overline{w}=\bs{w}\cdot \bs p,
\end{equation}
Here $S_l$ is the unit simplex
$$
S_l=\{\bs p\in \R^l\mid \bs p\geq 0,\,\sum_{i=0}^{l-1}p_i=1\},
$$and the dot denotes the standard inner product in $\R^l$.

If the vector (fitness landscape) $\bs w$ is fixed then $\bs p=\bs p(q)$ and $\overline{w}=\overline{w}(q)$ are functions of one parameter $q$, and our goal is to analyze behavior of $\bs p$ and $\overline{w}$ with respect to small perturbations of $q$. In particular, we are interested in finding the absolute minimum of $\overline{w}(q)$, which always exists and is determined by the condition $\overline{w}'(q)=0$. The main tool for our analysis is to first consider the adjoint eigenvalue problem (we call $\bs A^\top \bs y=\lambda \bs y$ the adjoint eigenvalue problem for $\bs A\bs x=\lambda \bs x$), and also perform the computations in the basis of the eigenvectors of $\bs Q$, similar to how we treated one of the close relatives of the Eigen model --- a permutation invariant Crow--Kimura model --- in \cite{bratus2013linear}.

To find $\overline{w}'(q)=\frac{\D \overline{w}}{\D q}(q)$ consider the adjoint eigenvalue problem (we have $(\bs{QW})^\top=\bs{WQ}$)
\begin{equation}\label{eq3:3}
    \bs{WQr}=\overline{w} \bs r,
\end{equation}
where $\bs r=(r_0,\ldots,r_{l-1})^\top$ is an eigenvector of $\bs{WQ}$ corresponding to the eigenvalue $\overline{w}$ for which the normalization
\begin{equation}\label{eq3:4}
    \bs p\cdot \bs r=\sum_{i=0}^{l-1}p_ir_i=1
\end{equation}
holds. Due to normalization \eqref{eq3:4}, we have
$$
\overline{w}=\bs r\cdot\bs{QWp},
$$
which yields
\begin{equation}\label{eq3:5}
    \overline{w}'=\bs r\cdot \bs Q'\bs{Wp}.
\end{equation}
Indeed,
$$
\overline{w}'=\bs p'\cdot \bs{WQr}+\bs p\cdot\bs{WQ}'\bs r+\bs p\cdot \bs{QWr}',
$$
and
$$
\bs p'\cdot \bs{WQr}+\bs p\cdot \bs{QWr}'=\overline{w}\bigl(\bs p'\cdot \bs r+\bs p\cdot \bs r'\bigl)=\overline{w}(\bs p\cdot \bs r)'=0.
$$

Direct calculations yield, for $q\neq0,1$,
\begin{equation}\label{eq3:6}
    \bs Q'=\frac{N}{q}\bs Q-\frac{1}{q(1-q)}\bs{S},
\end{equation}
where the matrix $\bs S=(s_{ij})_{l\times l}$ is defined as follows: $\bs S:=(H_{ij}q_{ij})_{l\times l},$ and $H_{ij}$ are the corresponding Hamming distances. On putting \eqref{eq3:6} into \eqref{eq3:5} we obtain
\begin{align*}
\overline{w}'&=\bs r\cdot \bs Q'\bs W\bs p=\frac{N}{q}\bs r\cdot \bs{QWp}-\frac{1}{q(1-q)}\bs r\cdot \bs{SWp}\\
&=\frac{N\overline{w}}{q}-\frac{1}{q(1-q)}\bs r\cdot \bs{SWp}.
\end{align*}
From \eqref{eq3:1} and the fact that $\bs Q$ is invertible (if $2q\neq 1$), one has $\bs{Wp}=\overline{w}\bs Q^{-1}\bs p$. Therefore,
\begin{equation}\label{eq3:7}
   \overline{w}'=\frac{N\overline{w}}{q}-\frac{\overline{w}}{q(1-q)}\bs r\cdot \bs{S}\bs Q^{-1}\bs p.
\end{equation}
Hence the condition $\overline{w}'(q)=0$ implies
\begin{equation}\label{eq3:8}
    N(1-q)=\bs r\cdot \bs S\bs Q^{-1}\bs p.
\end{equation}

Now let us rewrite problem \eqref{eq3:1} in the basis of the eigenvectors of $\bs Q$. The matrix of the change of basis is given by $\bs T$, where the $i$-th column of $\bs T$ is the $i$-th eigenvector of $\bs Q$ corresponding to the eigenvalue $\lambda_i=(2q-1)^{\Ham{i}}$, $\Ham{i}$ is the Hamming norm of the sequence type $i$ (see Section \ref{propQ} for details). We have
$$
\bs T^{-1}\bs{QWp}=\overline{w} \bs T^{-1}\bs p,
$$
or
$$
\bs T^{-1}\bs{Q}\bs T \bs T^{-1}\bs W\bs T \bs T^{-1}\bs p=\overline{w} \bs T^{-1}\bs p.
$$
Now, using the fact that $\bs T^{-1}\bs Q\bs T=\bs D=\diag\bigl(1,2q-1,\ldots,(2q-1)^{\Ham{i}},\ldots,(2q-1)^N\bigr)$ (note that $\bs D$ is an $l\times l$ matrix, for the exact form see Section \ref{propQ}) and notations
$$
\bs F:=\bs T^{-1}\bs W\bs T=2^{-N}\bs{TWT},\quad \bs x:=\bs T^{-1}\bs p,
$$
we rewrite our original problem \eqref{eq3:1} in the form
\begin{equation}\label{eq3:9}
    \bs{DFx}=\overline{w}\bs x.
\end{equation}
The normalization condition for $\bs x$ reads
$$
x_0=\frac{1}{2^N}\sum_{i=0}^{l-1}p_i=\frac{1}{2^N}\,.
$$

Again, consider the adjoint problem for \eqref{eq3:9} (now $(\bs{DF})^\top=\bs{FD}$):
\begin{equation}\label{eq3:10}
    \bs{FDy}=\overline{w} \bs y,
\end{equation}
for $\bs y=(y_0,\ldots,y_{l-1})^\top$ such that $\bs x\cdot \bs y=1$. We have
$$
\overline{w}=\bs y\cdot \bs{DFx},
$$
and hence
$$
\overline{w}'=\bs y\cdot\bs D'\bs{Fx}.
$$
Since
$$
\bs D'=2\diag(0,1,\ldots,\Ham{i}(2q-1)^{\Ham{i}-1},\ldots,N(2q-1)^{N-1}),
$$
we can, for $2q\neq 1$, write $\bs{Fx}=\overline{w} \bs D^{-1}\bs x$, and therefore obtain
$$
\overline{w}'=\bs y\cdot\bs D'\bs D^{-1}\overline{w}\bs x=\frac{2\overline{w}}{2q-1}\bs y\cdot \diag(0,1,\ldots,\Ham{i},\ldots,N)\bs x,
$$
or, in coordinates,
\begin{equation}\label{eq3:11}
    \left(q-\frac 12\right)\frac{\overline{w}'}{\overline{w}}=\sum_{i=0}^{l-1}\Ham{i}x_iy_i.
\end{equation}
Therefore the condition $\overline{w}'(q)=0$ implies
\begin{equation}\label{eq3:12}
    \sum_{i=0}^{l-1}\Ham{i}x_iy_i=0.
\end{equation}
Both conditions \eqref{eq3:8} and \eqref{eq3:12} can be used for search of the minimum of $\overline{w}(q)$.

\subsection{Case $q=0.5$}Here we show that for $q=0.5$ it is always true that $\overline{w}'(0.5)\geq 0$.

Direct calculations show that if $q=0.5$ then
$$
\bs Q=2^{-N}\bs E,\quad \bs Q'=2^{1-N}(N\bs E-2\bs H),
$$
where $\bs E$ is the matrix with all the entries equal to 1 and $\bs H=(H_{ij})$ is the matrix of Hamming distances.

In this case we can explicitly find the vectors $\bs p$ and $\bs r$ in \eqref{eq3:5}. Indeed, it can be checked that
\begin{equation}\label{eq4:2}
    \overline{w}(0.5)=2^{-N}\sum_{i=0}^{l-1}w_i,\quad \bs p^\top=2^{-N}(1,\ldots,1),\quad \bs r^\top=\frac{1}{\overline{w}(0.5)}\bs w^\top.
\end{equation}
Therefore, \eqref{eq3:5} implies
\begin{align*}
\overline{w}'(0.5)&=\bs r\cdot \bs Q'\bs W\bs p=\frac{1}{\overline{w}(0.5)2^{N-1}2^N}\,\bs w\cdot (N\bs E-2\bs H)\bs W\bs 1\\
&=\frac{1}{\overline{w}(0.5)2^{2N-1}}\,\bs w^\top(N\bs E-2\bs H) \bs w,
\end{align*}
where $\bs 1=(1,\ldots,1)^\top\in\R^l$.

Consider the quadratic form
$$
f\colon \bs w\to \bs w^\top(N\bs E-2\bs H) \bs w.
$$
This quadratic form is nonnegative. Indeed, it can be proved by induction that matrix $(N\bs E-2\bs H)$ has $N$ positive eigenvalues, each equal to $2^N$, and $2^N-N$ zero eigenvalues. Therefore, we have
\begin{proposition}\label{prop4:3}
\begin{equation}\label{eq3:5a}
    \overline{w}'(0.5)=\frac{1}{\overline{w}(0.5)2^{2N-1}}f(\bs w)\geq 0.
\end{equation}
\end{proposition}
A simple corollary to the last proposition is that the conditions
$$
\overline{w}'(0.5)=0
$$
and
$$
(N\bs E-2\bs H) \bs w=0
$$
are equivalent. If $(N\bs E-2\bs H) \bs w\neq 0$ then $\overline{w}'(0.5)>0$.

\begin{remark}It can be checked that the sum of the elements in every row and column of $(N\bs E-2\bs H)$ is zero. This implies that condition $\overline{w}'(0.5)=0$ is true for any ``symmetric'' vectors $\bs w$, which have $w_i=w_{i^\ast}$ for all coordinates, $i^\ast=2^N-1-i$.
\end{remark}

\subsection{On the minimum of $\overline{w}(q)$}
Here we assume that the fitness landscape $\bs W$ is such that the eigenvalues, which can be directly calculated, $\overline{w}(0)=\max\{\sqrt{w_i w_{i^\ast}}\}$ and $\overline{w}(1)=\max\{w_i\}$ have multiplicity one. In this case from \eqref{eq3:5} it follows that (see also Section \ref{sec:4:6})
\begin{equation}\label{eq4:1}
    \overline{w}'(0)=-N\overline{w}(0)<0,\quad \overline{w}'(1)=N\overline{w}(1)>0.
\end{equation}
This means that for $q\in[0,1]$ function $\overline{w}(q)$ decreases in the right neighborhood of zero and increases in the left neighborhood of 1. This, together with the continuity of $\overline{w}(q)$, implies that there exists an absolute minimum
$$
\hat{w}:=\mathrm{absmin}\,\{\overline{w}(q)\mid 0<q<1\}=\overline{w}(\hat{q}),
$$
and the coordinate $\hat{q}$ of this minimum is determined from $\overline{w}'(\hat{q})=0$. The analysis in Section \ref{sec:4:2} shows that this condition implies \eqref{eq3:8} and \eqref{eq3:12}, which can used to computations (another approach, which is convenient for small values of $N$, is outlined in Appendix \ref{append:B}).

\subsection{Analysis of the second derivative of $\overline{w}(q)$}\label{sec:4:5}
In this subsection we continue the general analysis of $\overline{w}(q)$. In particular we concentrate on the second derivative of this function and show that it is convex when $q>0.5$, provided some mild conditions on $\bs w$ are satisfied.

Together with the fitness landscape
$$
\bs W=\diag(w_0,\ldots,w_{l-1})
$$
consider a diagonal matrix
$$
\sqrt\bs{{W}}:=\diag (\sqrt{w_0},\ldots,\sqrt{w_{l-1}}).
$$
Then
$$
\bs{QW}=\sqrt{\bs W}^{-1}\sqrt{\bs W}\bs{Q}\sqrt{\bs W}\sqrt{\bs W},
$$
i.e., matrix $\bs{QW}$ is similar to the symmetric matrix $\sqrt{\bs W}\bs{Q}\sqrt{\bs W}$, and hence its eigenvalues are all real, and the eigenvectors form a basis of $\R^l$ (this observation was made in~\cite{rumschitzki1987spectral}).

Recall from Section \ref{sec:4:2} that together with the matrix $\bs{QW}$ we consider matrix $\bs{DF}$, for which we define
$$
\sqrt{\bs F}:=\bs T^{-1}\sqrt{\bs W}\bs T.
$$
A symmetric matrix
\begin{equation}\label{eq7:1}
    \bs L:=\sqrt{\bs F}\bs D\sqrt{\bs F}
\end{equation}
is similar to $\bs{DF}$, and hence to $\bs{QW}$, and has the same eigenvalues, the maximal of which, $\overline{w}(q)$, is of primary interest to us.

Let $\bs z=\bs z(q)$ be the eigenvector of $\bs L$ corresponding to $\overline{w}(q)$ satisfying the condition $\bs z\cdot \bs z=1$. Note that $\bs z=\sqrt{\bs F}\bs x=\sqrt{\bs F}\,\bs T^{-1}\bs p$.  Due to the normalization,
\begin{equation}\label{eq7:2}
    \overline{w}=\bs z\cdot \bs{Lz}.
\end{equation}
By differentiating \eqref{eq7:2} with respect to $q$, we obtain
$$
\overline{w}'=\bs z'\cdot \bs{Lz}+\bs z\cdot \bs L'\bs z+\bs z\cdot \bs L\bs z'.
$$
Noting that
$$
\bs z'\cdot \bs{Lz}+\bs z\cdot \bs L\bs z'=\overline{w}(\bs z'\cdot \bs z+\bs z\cdot \bs z')=\overline{w}(\bs z\cdot \bs z)'=0,
$$
we get
\begin{equation}\label{eq7:3}
\overline{w}'=\bs z\cdot \bs{L}'\bs z.
\end{equation}
Now we differentiate \eqref{eq7:3} with respect to $q$:
$$
\bs{w}''=\bs z'\cdot\bs{L}'\bs z+\bs z\cdot \bs{L}''\bs z+\bs z\cdot\bs{L}'\bs z'=2\bs z'\cdot\bs{L}'\bs z+\bs z\cdot \bs{L}''\bs z.
$$
$\bs{Lz}=\overline{w}\bs z$ implies $\bs L'\bs z+\bs L\bs z'=\overline{w}'\bs z+\overline{w}\bs z',$ therefore, by expressing $\bs L'\bs z$ in the last equality, we get
$$
\bs{w}''=2 \bs z'\cdot(\overline{w}\bs I-\bs L)\bs z'+\bs z\cdot \bs{L}''\bs z,
$$
since $2\bs z'\cdot \bs z=(\bs z\cdot \bs z)'=0$.

Due to the fact that $\bs L$ is symmetric with the maximal eigenvalue $\overline{w}$, matrix $\overline{w}\bs I-\bs L$ is nonnegative definite, and hence the quadratic form
$$
f_1\colon \bs v\to 2 \bs v\cdot(\overline{w}\bs I-\bs L)\bs v
$$
is nonnegative on any vector $\bs v\in \R^l$. In particular, $f_1(\bs z')\geq 0$.

Now consider $\bs z\cdot \bs L''\bs z$. Since $\bs F$ is constant, then $\bs L''=\sqrt{\bs F}\bs D''\sqrt{\bs F}$. Using the explicit form of $\bs D$, we find first
$$
\bs D'=\frac{2}{2q-1}\bs D\Delta,
$$
where $\Delta=\diag(0,1,\ldots,H_i,\ldots,N)$ is a constant matrix. Furthermore,
$$
\bs D''=\frac{4}{(2q-1)^2}\bs D(\Delta^2-\Delta),
$$
from where it follows that the quadratic form
$$
f_2\colon \bs v\to\frac{4}{(2q-1)^2}\bs v\cdot\bs D(\Delta^2-\Delta)\bs v
$$
is nonnegative definite for $q>0.5$ (because $\bs D$ is positive definite and $\Delta^2-\Delta$ is nonnegative definite). Therefore,
$$
\bs z\cdot\bs{L''z}=f_2(\sqrt{\bs F}\bs z)\geq 0,
$$
and we obtain
\begin{proposition} If $q\geq 0.5$ then $\overline{w}(q)$ is nondecreasing and convex.
\end{proposition}
\begin{proof} Due to Proposition \ref{prop4:3} we have $\overline{w}'(0.5)\geq 0$. Since $\overline{w}''(q)\geq 0$ for $q\in(0.5,1]$ then $\overline{w}'(q)\geq\overline{w}'(0.5)\geq 0$. Therefore, $\overline{w}(q)$ is nondecreasing and convex.
\end{proof}
\begin{corollary}The absolute minimum $\hat{q}$ of $\overline{w}(q)$ satisfies the condition
$$
0\leq \hat{q}\leq 0.5.
$$
\end{corollary}
\begin{remark}The example with $\bs w=(1,2,2,1,2,1,1,2)$ shows that $\overline{w}(q)$ can change the convexity after passing $q=0.5$ (see Figure \ref{fig:5}).
\begin{figure}[!ht]
\centering
\includegraphics[width=0.45\textwidth]{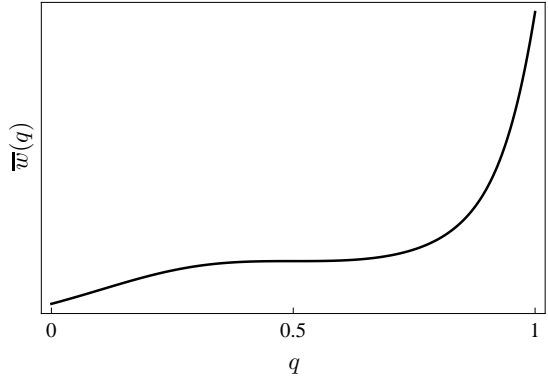}
\caption{An example of a fitness landscape when $\overline{w}''(q)<0$ for $q<0.5$}\label{fig:5}
\end{figure}
\end{remark}
Putting together the observations and analysis in Sections \ref{propQ}--\ref{sec:4:5}, we obtain a proof of Theorem~\ref{theor:main:1}.

In conclusion of this section we mention that another theoretical approach to analyze the second derivative of  $\overline{w}(q)$ is presented in Appendix \ref{append:C}.

\subsection{Calculations of $\overline{w}(q),\overline{w}'(q),\overline{w}''(q)$ for $q=0,0.5,1$}\label{sec:4:6}
In this section we present the calculations, based on the analysis from Sections \ref{propQ}--\ref{sec:4:5}, that lead to Table \ref{tab:main:1}. We also show that it is possible to find expressions for the second derivatives, that we did not included in Table \ref{tab:main:1}.

Recall that we consider the direct
\begin{equation*}
\bs{QWp}=\overline{w}\bs p,\tag{\ref{eq3:1}}
\end{equation*}
and adjoint
\begin{equation*}
\bs{WQr}=\overline{w}\bs r,\tag{\ref{eq3:3}}
\end{equation*}
eigenvalue problems, such that
$$
\bs p\cdot\bs r=1.
$$
We showed
\begin{equation*}
    \overline{w}'=\bs r\cdot \bs Q'\bs{Wp}.\tag{\ref{eq3:5}}
\end{equation*}
By differentiating \eqref{eq3:5}, \eqref{eq3:1}, and \eqref{eq3:3} (the last one with transposing), we have
\begin{align}
\overline{w}''=\bs r'\cdot \bs Q'\bs W\bs p&+\bs r\cdot \bs Q''\bs W\bs p+ \bs r\cdot \bs Q'\bs W\bs p',\label{eqN:3}\\
\bs Q'\bs W \bs p+\bs {QWp}'&=\overline{w}'\bs p+\overline{w}\bs p',\label{eqN:1}\\
\bs (r^\top)'\bs{QW}+\bs r^\top\bs Q'\bs W&=\overline{w}'\bs r^\top+\overline{w}(\bs r^\top)'.\label{eqN:2}
\end{align}
We first multiply \eqref{eqN:1} from the left by $(\bs r^\top)'$, and \eqref{eqN:2} from the right by $\bs p'$ and obtain
\begin{align*}
(\bs r^\top)'\bs Q'\bs W \bs p&=\overline{w}'(\bs r^\top)'\bs p+\overline{w}(\bs r^\top)'\bs p'-(\bs r^\top)'\bs {QWp}',\\
\bs (r^\top)'\bs{Q}'\bs{W}\bs p'&=\overline{w}'\bs r^\top\bs p'+\overline{w}(\bs r^\top)'\bs p'-\bs r^\top\bs Q\bs W\bs p'.
\end{align*}
Plugging these expressions into \eqref{eqN:3} and taking into account
$$
\overline{w}' \bs r'\cdot \bs p+\overline{w}'\bs r\cdot \bs p'=\overline{w}'(\bs r\cdot \bs p)'=0,
$$
we obtain an expression for the second derivative
\begin{equation}\label{eqN:4}
    \overline{w}''=\bs r'\cdot \bs{QWp}+2\bs r'\cdot (\overline{w}\bs I-\bs{QW})\bs p'.
\end{equation}

\paragraph{}\textit{Case $q=1$.} In this case we have
$$
\overline{w}(1)=\max\{w_i\}=w_k,
$$
and we only consider the case $w_i<w_k$ for $i\neq k$. The coordinates of the eigenvectors $\bs p$ and $\bs r$ in \eqref{eq3:1} and \eqref{eq3:3} respectively are found by direct calculations:
$$
p_i=r_i=\delta_{ik},
$$
where $\delta_{ik}$ is the Kronecker delta. Equation \eqref{eq3:5} yields then $\overline{w}'(1)=N\overline{w}(1)$.

To find the second derivative, we use \eqref{eqN:1} and \eqref{eqN:2} together with the conditions
$$
\sum_{i=0}^{l-1}p_i'=0,\quad \sum_{i=0}^{l-1}r_i'p_i+r_ip_i'=0.
$$
We find
$$
p_i'=\begin{cases}
0,&H_{ik}\geq 2,\\[4mm]
\displaystyle-\frac{w_k}{w_k-w_i}\,,& H_{ik}=1,\\[4mm]
\displaystyle\sum_{j\colon H_{jk}=1}\frac{w_k}{w_k-w_j}\,,&H_{ik}=0,
\end{cases}\qquad\text{and}\qquad
r_i'=\begin{cases}
0,&H_{ik}\geq 2,\\[4mm]
\displaystyle-\frac{w_i}{w_k-w_i}\,,& H_{ik}=1,\\[4mm]
\displaystyle-\sum_{j\colon H_{jk}=1}\frac{w_k}{w_k-w_j}\,,&H_{ik}=0.
\end{cases}
$$
After plugging these expressions into \eqref{eqN:4}, we find
$$
\overline{w}''(1)=N(N-1)\overline{w}(1)+2\overline{w}(1)S_k^{(2)}(\bs w),
$$
where
$$
S_k^{(2)}(\bs w)=\sum_{j\colon H_{jk}=1}\frac{w_j}{w_k-w_j}\,.
$$

\paragraph{}\textit{Case $q=0$.} In this case
$$
\overline{w}=\max\{\sqrt{w_iw_{i^\ast}}\}=\sqrt{w_kw_{k^\ast}},
$$
and we assume that $\sqrt{w_iw_{i^\ast}}<\sqrt{w_kw_{k^\ast}}$ for all $i\neq k,k^\ast$.

Here $i^\ast$ is the index conjugate to $i$, formally $i^\ast=2^N-1-i$. We have $(i^\ast)^\ast=i$, and the binary representations of the conjugate indices are related as
$$
i=[a_0,a_1,\ldots,a_{N-1}],\quad i^\ast=[a_0^\ast,a_1^\ast,\ldots,a_{N-1}^\ast], \quad a_k^\ast=1-a_k,\quad a_i,a_i^\ast\in\{0,1\}.
$$
The Hamming distance between two conjugate indices is maximal: $H_{ii^\ast}=N$.

Coordinates of the eigenvectors $\bs p$ and $\bs r$  are found as
$$
p_i=\begin{cases}
0,&i\neq k,k^\ast,\\[3mm]
\displaystyle\frac{\sqrt{w_{k^\ast}}}{\sqrt{w_k}+\sqrt{w_{k^\ast}}}\,,&i=k,\\[3mm]
\displaystyle\frac{\sqrt{w_{k}}}{\sqrt{w_k}+\sqrt{w_{k^\ast}}}\,,&i=k^\ast,
\end{cases}\qquad\text{and}\qquad
r_i=\begin{cases}
0,&i\neq k,k^\ast,\\[3mm]
\displaystyle\frac{\sqrt{w_k}+\sqrt{w_{k^\ast}}}{2\sqrt{w_{k^\ast}}}\,,&i=k,\\[3mm]
\displaystyle\frac{\sqrt{w_k}+\sqrt{w_{k^\ast}}}{2\sqrt{w_k}}\,,&i=k^\ast.
\end{cases}
$$
From \eqref{eq3:5} it then follows that
$$
\overline{w}'(0)=-N\overline{w}(0).
$$

Similarly to the case $q=1$, for $N\geq 3$, we find for the case $H_{ik}=1$:
\begin{align*}
p_i'&=\frac{\sqrt{w_kw_{k^\ast}}w_{k^\ast}p_{k^\ast}+w_{i^\ast}w_kp_k}{w_kw_{k^\ast}-w_iw_{i^\ast}}\,,\quad p_{i^\ast}'=\frac{\sqrt{w_kw_{k^\ast}}w_{k}p_{k}+w_{i}w_{k^\ast}p_{k^\ast}}{w_kw_{k^\ast}-w_iw_{i^\ast}}\,,\\
r_i'&=\frac{\sqrt{w_kw_{k^\ast}}w_{i}r_{k}+w_{i}w_{i^\ast}r_{k^\ast}}{w_kw_{k^\ast}-w_iw_{i^\ast}}\,,\quad r_{i^\ast}'=\frac{\sqrt{w_kw_{k^\ast}}w_{i^\ast}r_{k^\ast}+w_{i^\ast}w_{i}r_{k}}{w_kw_{k^\ast}-w_iw_{i^\ast}}\,.
\end{align*}
The expressions for $p'_k,p'_{k^\ast},r_k',r_{k^\ast}'$ are found from the system of equations
\begin{align*}
p'_k+p'_{k^\ast}&=-\sum_{i\colon H_{ik}=1}(p'_i+p'_{i^\ast}),\quad \sqrt{w_k}p'_k=\sqrt{w_{k^\ast}}p'_{k^\ast},\\
r'_kp_k+r'_{k^\ast}p_{k^\ast} &=- r_kp'_k+r_{k^\ast}p'_{k^\ast},\quad \sqrt{w_{k^\ast}}r'_k=\sqrt{w_{k}}r'_{k^\ast}.
\end{align*}
For all other indices $j$ components $p'_j,r_j'$ are equal to zero.

Plugging the found values into \eqref{eqN:4} and after canceling and simplifications, we obtain
\begin{equation}\label{eqN:5}
\overline{w}''(0)=N(N-3)\overline{w}(0)+2\overline{w}^3(0)S_k^{(0)}(\bs w)+\overline{w}(0)S_k^{(1)}(\bs w),
\end{equation}
where
\begin{align*}
S_k^{(0)}(\bs w)&=\sum_{i\colon H_{ik}=1}\frac{1}{w_kw_{k^\ast}-w_iw_{i^\ast}}\,,\\
S_k^{(1)}(\bs w)&=w_{k^\ast}\sum_{i\colon H_{ik}=1}\frac{w_i}{w_kw_{k^\ast}-w_iw_{i^\ast}}+w_k\sum_{i\colon H_{ik}=1}\frac{w_{k^\ast}}{w_kw_{k^\ast}-w_iw_{i^\ast}}\,.
\end{align*}

\paragraph{}\textit{Case $q=0.5$.} In this case direct calculations lead to
$$
\overline{w}(0.5)=\frac{1}{2^N}\sum_{i=0}^{l-1}w_i,\quad \bs p^\top=\frac{1}{2^N}(1,\ldots,1),\quad \bs r^\top=\frac{1}{\overline{w}(0.5)}(w_0,w_1,\ldots,w_{l-1})=\frac{1}{\overline{w}(0.5)}\bs w^\top.
$$
Equation \eqref{eq3:5} implies
$$
\overline{w}'=\bs r\cdot \bs Q'\bs W\bs p=\frac{1}{2^{2N-1}\overline{w}(0.5)}\bs w\cdot \bs S\bs w,
$$
where $\bs S=N\bs E-2\bs H,$ matrix $\bs E$ is composed of all 1s, and $\bs H=(H_{ij})_{l\times l}$ is the matrix of the Hamming distances.

Consider \eqref{eqN:1}. If $q=0.5$ then it can be calculated that $\bs{QWp}'=\overline{w}'\bs p$. Therefore, $\bs Q'\bs W\bs p=\overline{w}\bs p$, or
$$
\bs p'=\frac{1}{\overline{w}}\bs Q'\bs W\bs p.
$$
Similarly, for $q=0.5$,
$$
\bs{QWp}''=\overline{w}''\bs p.
$$
The last two expressions can be checked with
$$
\overline{w}'=\bs w\cdot\bs p',\quad \overline{w}''=\bs w\cdot \bs p''.
$$
Now differentiate \eqref{eqN:1} with respect to $q$:
$$
\bs Q''\bs{Wp}+2\bs Q'\bs{Wp}'+\bs{QWp}''=\overline{w}''\bs p+2\overline{w}'\bs p'+\overline{w}\bs p'',
$$
or, taking into account the previous equalities,
$$
\bs Q''\bs{Wp}+2\bs Q'\bs{Wp}'=\overline{w}'\bs p'+\overline{w}\bs p''.
$$
From the last expression we can find $\bs p''$ and plug it into $\overline{w}''=\bs w\cdot\bs p''$:
$$
\overline{w}''=\frac{1}{\overline{w}}\bs w\cdot \bs Q''\bs W\bs p+\frac{2}{\overline{w}}\bs w\cdot \bs Q'\bs W\bs p'-\frac{2}{\overline{w}}\bs w\cdot \overline{w}'\bs p',
$$
or, using $\overline{w}'=\bs w\cdot \bs p'$,
$$
\overline{w}''=\frac{1}{\overline{w}}\bs w\cdot \bs Q''\bs W\bs p+\frac{2}{\overline{w}}\bs w\cdot \bs Q'\bs W\bs p'-\frac{2(\overline{w}')^2}{\overline{w}}\,.
$$
If we now use
$$
\bs Q'=2^{1-N}\bs S,\quad \bs Q''=2^{2-N}\bs B,
$$
where $\bs S=N\bs E-2\bs H$, $\bs B=N^2\bs E-N\bs E-4N\bs H+\bs{H^2}$, and $\bs{H^2}=(H_{ij}^2)_{l\times l}$, then we can compute
\begin{equation}\label{eq:N:6}
\overline{w}''(0.5)=\frac{S^{(3)}(\bs w)}{\overline{w}(0.5)2^{2N-4}}+\frac{S^{(4)}(\bs w)}{\overline{w}^2(0.5)2^{3N-5}}-4N(3N+1)\overline{w}(0.5)+12N\overline{w}'(0.5)-\frac{2\bigl(\overline{w}'(0.5)\bigr)^2}{\overline{w}(0.5)}\,,
\end{equation}
where
$$
S^{(3)}(\bs w)=\bs w\cdot \bs{H^2}\bs w
$$
is a quadratic form with matrix $\bs{H^2}$, and
$$
S^{(4)}(\bs w)=\bs w\cdot \bs H\bs W\bs H\bs w,
$$
is a cubic form of the vector $\bs w$.

Putting all the calculations together we obtain Table \ref{tab:N:1}.
\begin{table}[!ht]
  \centering
  \begin{tabular}{| c | c | c | c |}
  \hline&&&\\[-3mm]
  $q$&$0$&$0.5$&$1$\\[1mm]
  \hline&&&\\[-3mm]
  $\overline{w}(q)$&$\max \{\sqrt{w_iw_{i^\ast}}\}$&$\displaystyle\frac{1}{2^N}\sum_{i=0}^{2^N-1}w_i$&$\max\{w_i\}$\\[5mm]
  \hline&&&\\[-3mm]
  $\overline{w}'(q)$ & $-N\max\{\sqrt{w_i w_{i^\ast}}\}$ & $\displaystyle \frac{1}{\overline{w}(0.5)2^{2N-1}}\,\bs w\cdot \bs A\bs w$ & $N\max\{w_i\}$\\[4mm]
  \hline&&&\\[-3mm]
 $\overline{w}''(q)$ &  \eqref{eqN:5} & \eqref{eq:N:6} & $N(N-1)\overline{w}(1)+2\overline{w}(1)S_k^{(2)}(\bs w)$\\[2mm]
 \hline
\end{tabular}
  \caption{Expressions for $\overline{w}(q)$ for particular $q$. Here $i^\ast$ is the index conjugate to $i$, formally, $i^\ast=2^N-1-i$. Matrix $\bs A=N\bs E-2\bs H$, where $\bs E$ is the matrix of all 1s, $\bs H$ is the matrix with elements $H_{ij}$, where $H_{ij}$ is the Hamming distance between sequences $i$ and $j$, and the expressions for $S_k^{(2)}$ is defined in the text}\label{tab:N:1}
\end{table}

\section{Single peaked landscape}\label{sec:5}
Here we apply the results from Section \ref{sec:4} to one specific example of the single peaked landscape.
\subsection{General analysis}
Consider the fitness landscape
\begin{equation}\label{eq5:1}
    \bs w=(w,\ldots,w,w+s,w,\ldots,w),\quad w,s>0,
\end{equation}
where $w+s$ is at the $k$-th place.
The corresponding matrix
$$
\bs W=w\bs I+s\bs E_{k,k},
$$
where $\bs E_{k,k}$ is the matrix with 1 at the intersection of the $k$-th row and the $k$-th column and 0s everywhere else, and $w>0,s>0$. Using Table \ref{tab:N:1} with \eqref{eq5:1}, we can calculate $\overline{w}(q),\overline{w}'(q),\overline{w}''(q)$ for \eqref{eq5:1}, the results are presented in Table \ref{tab:main:2}.

For $q=1$ the eigenvector (quasispecies) $\bs p$ is
$$
p_k(1)=1,\quad p_i(1)=0,\quad i\neq k.
$$
For $q=0$ the eigenvector $\bs p$ is
$$
p_k(0)=\frac{\sqrt{w}}{\sqrt{w+s}+\sqrt{w}}\,,\quad p_{k^\ast}=\frac{\sqrt{w+s}}{\sqrt{w+s}+\sqrt{w}}\,,\quad p_i=0,\,i\neq k,k^\ast.
$$

Using \eqref{eq5:1} in $\bs{QWp}=\overline{w} \bs p$ we find
$$
w\bs Q\bs p+ s\bs Q\bs E_{k,k}\bs p=\overline{w}\bs p.
$$
We also have
$$
\overline{w}=\sum_{i=0}^{l-1}w_ip_i=w+sp_k,
$$
or
$$
\overline{w}-w=sp_k.
$$
The last expression implies that our eigenvalue problem can be rewritten as
\begin{equation}\label{eq5:2}
w\bs Q\bs p+(\overline{w}-w)\bs Q_k=\overline{w}\bs p,
\end{equation}
where $\bs Q_k$ is the $k$-th column of $\bs Q$. Now we rewrite problem \eqref{eq5:2} in the coordinates of the eigenvectors of $\bs Q$. Recall that $\bs T$ is the matrix composed of the eigenvectors of $\bs Q$,
$$
\bs T^{-1}\bs{QT}=\bs D,\quad \bs x=\bs T^{-1}\bs p.
$$

We have from \eqref{eq5:2}
$$
w\bs{T}^{-1}\bs Q\bs{TT}^{-1}\bs p+(\overline{w}-w)\bs T^{-1}\bs Q_k=\overline{w}\bs{T}^{-1}\bs p,
$$
therefore,
\begin{equation}\label{eq5:3}
w\bs D\bs x+\frac{\overline{w}-w}{2^{N}}\bs{DT}_k=\overline{w}\bs x,
\end{equation}
where $\bs T_k$ is the $k$-th column of $\bs T$. The last equality holds because of
$$
\bs T^{-1}\bs Q_k=\bs{T}^{-1}\bs Q\bs e_k=\bs{DT}^{-1}\bs e_k=2^{-N}\bs{DT}\bs e_k=2^{-N}\bs{DT}_k.
$$
Equation \eqref{eq5:3} in coordinates is
$$
w(2q-1)^{\Ham{i}}x_i+\frac{\overline{w}-w}{2^N}(2q-1)^{\Ham{i}}t_{ik}=\overline{w}x_i,\quad i=0,\ldots,l-1,
$$
from where
$$
x_i=\frac{1}{2^N}\frac{(\overline{w}-w)(2q-1)^{\Ham{i}}}{\overline{w}-w(2q-1)^{\Ham{i}}}\,t_{ik},\quad i=0,\ldots,l-1.
$$
We note that for $i=0$ we have, as expected, $x_0=2^{-N}$. Since $\bs p=\bs{Tx}$, we find
\begin{equation}\label{eq5:4}
    p_j=\sum_{i=0}^{l-1}t_{ji} x_i=\frac{1}{2^N}\sum_{i=0}^{l-1}\frac{(\overline{w}-w)(2q-1)^{\Ham{i}}}{\overline{w}-w(2q-1)^{\Ham{i}}}\,t_{ji}t_{ik},\quad j=0,\ldots,l-1.
\end{equation}
Expressions \eqref{eq5:4} allows to determine the quasispecies distribution (the eigenvector $\bs p$) if the eigenvalue $\overline{w}(q)$ is known.

If $j=k$ in \eqref{eq5:4} then this expression simplifies since $t_{ki}t_{ik}=(\pm 1)^2=1$:
$$
p_k=\frac{1}{2^N}\sum_{i=0}^{l-1}\frac{(\overline{w}-w)(2q-1)^{\Ham{i}}}{\overline{w}-w(2q-1)^{\Ham{i}}}\,.
$$
Using the fact that $\overline{w}-w=sp_k$, we obtain
$$
\overline{w}-w=\frac{s}{2^N}\sum_{i=0}^{l-1}\frac{(\overline{w}-w)(2q-1)^{\Ham{i}}}{\overline{w}-w(2q-1)^{\Ham{i}}}\,.
$$

The last equality allows to make several conclusions. First, $\overline{w}=\overline{w}(q)$ is a root of the equation
\begin{equation}\label{eq5:5}
1=\frac{s}{2^N}\sum_{i=0}^{l-1}\frac{(2q-1)^{\Ham{i}}}{\overline{w}-w(2q-1)^{\Ham{i}}}\,.
\end{equation}
Second, by differentiating the last equality with respect to $q$ we find
$$
\overline{w}\sum_{i=1}^{l-1}\frac{2\Ham{i}(2q-1)^{\Ham{i}-1}}{(\overline{w}-s(2q-1)^{\Ham{i}})^2}=\overline{w}'\sum_{i=0}^{l-1}\frac{(2q-1)^{\Ham{i}}}{(\overline{w}-s(2q-1)^{\Ham{i}})^2}\,.
$$
If for $q_0$ it holds that $\overline{w}'(q_0)=0$, then
$$
\sum_{i=1}^{l-1}\frac{2\Ham{i}(2q-1)^{\Ham{i}-1}}{(\overline{w}-s(2q-1)^{\Ham{i}})^2}=0.
$$
Due to the fact that for $0.5\leq q_0\leq 1$ the right hand side of the last equality is positive, we independently of Theorem \ref{theor:main:1} proved
\begin{proposition}In the case of the single peaked landscape \eqref{eq5:1} the value of $q$ for which $\overline{w}(q)$ is minimal satisfies
$$
0<q<0.5.
$$
For $q\geq 0.5$ we have that $\overline{w}'(q)>0$.
\end{proposition}
Of course, the last proposition is a simple corollary of the general Theorem \ref{theor:main:1}. Which is more important here, is that expression \eqref{eq5:5} can be used to find simple analytical expressions for upper and lower bounds on $\overline{w}(q)$ for $0\leq q\leq 1$ for the single peaked landscape.

\subsection{On the bounds for $\overline{w}(q)$ in the case of the single peaked landscape}Here we consider the case when the fittest sequence has type 0, i.e., consists of all 0s:
$$
\bs{w}^\top=(w+s,w,\ldots,w).
$$
Rewrite equation \eqref{eq5:5} in the form
\begin{equation}\label{eq6:1}
\frac{s}{2^N}\sum_{k=0}^{N}\frac{\binom{N}{k}(2q-1)^k}{\overline{w}-w(2q-1)^{k}}=1.
\end{equation}
For the minimal value of the mean fitness $\hat{w}=\overline{w}(\hat{q})$ it holds that $\hat{w}>w$ due to
$$
\overline{w}=w+sp_0\geq w,
$$
and the equality $\hat{w}=w$ would contradict \eqref{eq6:1} (the term for $k=0$ would tend to infinity). This means that for all $q\in[0,1]$
$$
0<\frac{w}{\overline{w}}\leq \frac{w}{\hat{w}}<1.
$$

Introduce the notation $t=2q-1$, therefore $t\in[-1,1]$. For this new variable $t$, which is more convenient for the subsequent computations, \eqref{eq6:1} takes the form
\begin{equation}\label{eq6:2}
    \frac{s}{\overline{w}}\sum_{k=0}^N\frac{\binom{N}{k}}{2^N}\cdot \frac{t^k}{1-\frac{w}{\overline{w}}t^k}=\frac{s}{\overline{w}}\sum_{k=0}^Nt^k\sum_{j=0}^\infty \left(\frac{w}{\overline{w}t^k}\right)^j\frac{\binom{N}{k}}{2^N}=1,
\end{equation}
or, exchanging the order of the sums (this can be done due to the absolute convergence of all the series involved):
$$
\sum_{j=0}^\infty \left(\frac{w}{\overline{w}}\right)^{j+1}\sum_{k=0}^N\frac{\binom{N}{k}}{2^N}\left(t^{j+1}\right)^k=\frac{w}{s}\,.
$$
The binomial formula yields
$$
\sum_{k=0}^N\frac{\binom{N}{k}}{2^N}\left(t^{j+1}\right)^k=\left(\frac{1+t^{j+1}}{2}\right)^N.
$$
Introducing $m=j+1$, we finally obtain
\begin{equation}\label{eq6:3}
    \sum_{m=1}^\infty \left(\frac{w}{\overline{w}}\right)^m\left(\frac{1+t^{m}}{2}\right)^N=\frac{w}{s}\,,\quad t=2q-1,\,-1\leq t\leq 1.
\end{equation}

First consider the case $-1\leq t\leq 0$, which corresponds to $0\leq q\leq 0.5$. Recall that the minimum value of $\overline{w}$ is attained in this interval. The following estimates hold in this interval:
\begin{align*}
\left(\frac{1+t^{m}}{2}\right)^N &\geq \left(\frac{1-|t|}{2}\right)^N\geq 0\\
0\leq \left(\frac{1+t^{m}}{2}\right)^N&\leq \left(\frac{1+|t|}{2}\right)^N,\quad \text{if }m\text{ is even,}\\
0\leq \left(\frac{1+t^{m}}{2}\right)^N&\leq \left(\frac 12\right)^N,\quad \text{if }m\text{ is odd,}
\end{align*}
The first of these inequalities together with \eqref{eq6:3} imply
$$
\frac{w}{s}\geq \left(\frac{1-|t|}{2}\right)^N\sum_{m=1}^\infty\left(\frac{w}{\overline{w}}\right)^m=\left(\frac{1-|t|}{2}\right)^N \frac{w}{\overline{w}-w}\,,
$$
from where a lower bound
\begin{equation}\label{eq6:4}
    \overline{w}\geq w+s\left(\frac{1-|t|}{2}\right)^N
\end{equation}
follows.

On the other hand, using two other inequalities, we get
$$
\frac{w}{s}\leq \left(\frac{1+|t|}{2}\right)^N \sum_{j=1}^\infty \left(\frac{w}{\overline{w}}\right)^{2j}+\left(\frac 12\right)^N \sum_{j=1}^\infty\left(\frac{w}{\overline{w}}\right)^{2j-1},
$$
or
$$
\frac{w}{s}\leq \left(\frac{1+|t|}{2}\right)^N \frac{w^2}{\overline{w}^2-w^2}+\left(\frac 12\right)^N \frac{w\overline{w}}{\overline{w}^2-w^2}\,.
$$
By solving this quadratic inequality we find an upper bound on $\overline{w}$:
\begin{equation}\label{eq6:5}
    \overline{w}\leq \sqrt{w^2+sw\left(\frac{1+|t|}{2}\right)^N+\left(\frac{s}{2^{N+1}}\right)^2}+\frac{s}{2^{N+1}}\,.
\end{equation}
Returning to the original mutation probability $q$, we obtain
\begin{proposition}\label{prop:5:2}Consider the Eigen quasispecies model with the single peaked landscape $\bs w^\top=(w+s,w,\ldots,w)$ for $w,s>0$. If $0\leq q\leq 0.5$ then the mean fitness $\overline{w}(q)$ satisfies
\begin{equation}\label{eq6:6}
    w+sq^N\leq \overline{w}(q)\leq \sqrt{w^2+sw(1-q)^N+s^22^{-2(N+1)}}+s2^{-N-1}.
\end{equation}
\end{proposition}
Now consider $0\leq t\leq 1$. Here we have
\begin{align*}
\left(\frac{1+t^{m}}{2}\right)^N &\geq \left(\frac 12\right)^N\geq 0,\\
0\leq \left(\frac{1+t^{m}}{2}\right)^N &\leq \left(\frac{1+t^{2}}{2}\right)^N ,\quad\text{if }m\text{ is even,}\\
0\leq \left(\frac{1+t^{m}}{2}\right)^N &\leq \left(\frac{1+t}{2}\right)^N ,\quad\text{if }m\text{ is odd.}
\end{align*}
The first of these inequalities together with \eqref{eq6:3} imply
$$
\frac{w}{s}\geq \left(\frac{1}{2}\right)^N\sum_{m=1}^\infty\left(\frac{w}{\overline{w}}\right)^m= \left(\frac{1}{2}\right)^N \frac{w}{\overline{w}-w}\,.
$$
From where
$$
\overline{w}\geq w+s2^{-N}.
$$
On the other hand, two other inequalities and \eqref{eq6:3} imply
$$
\frac{w}{s}\leq \left(\frac{1+t^2}{2}\right)^N \sum_{j=1}^\infty \left(\frac{w}{\overline{w}}\right)^{2j}+\left(\frac{1+t}{2}\right)^N \sum_{j=1}^\infty\left(\frac{w}{\overline{w}}\right)^{2j-1},
$$
which yields
$$
\frac{w}{s}\leq \left(\frac{1+t^2}{2}\right)^N \frac{w^2}{\overline{w}^2-w^2}+\left(\frac{1+t}{2}\right)^N \frac{w\overline{w}}{\overline{w}^2-w^2}\,.
$$
By solving the quadratic inequality, we find
$$
\overline{w}\leq \sqrt{w^2+sw \left(\frac{1+t^2}{2}\right)^N+\frac{s^2}{4}\left(\frac{1+t}{2}\right)^{2N}}+\frac{s}{2}\left(\frac{1+t}{2}\right)^N.
$$
Returning to the original variable $q$, we conclude
\begin{proposition}\label{prop:5:3}Consider the Eigen quasispecies model with the single peaked landscape $\bs w^\top=(w+s,w,\ldots,w)$ for $w,s>0$. If $0.5\leq q\leq 1$ then the mean fitness $\overline{w}(q)$ satisfies
\begin{equation}\label{eq6:7}
    w+s2^{-N}\leq \overline{w}(q)\leq \sqrt{w^2+ sw(2q^2-2q+1)^N+\left(\frac{sq^N}{2}\right)^2}+\frac{sq^{N}}{2}\,.
\end{equation}
\end{proposition}

It is possible to obtain other lower and upper bounds for $\overline{w}(q)$ if $0.5\leq q\leq 1$. Note that $t^m$ is convex if $0\leq t\leq 1$. Then Jensen's inequality implies
$$
\frac{1+t^m}{2}\geq\left(\frac{1+t}{2}\right)^m=q^m.
$$
Therefore, from \eqref{eq6:3},
$$
\frac{w}{s}\geq \sum_{m=1}^\infty\left(\frac{wq^N}{\overline{w}}\right)=\frac{wq^N}{\overline{w}-wq^N}\,.
$$
From where we obtain a lower bound
\begin{equation}\label{eq6:8}
    \overline{w}(q)\geq (w+s)q^N, \quad 0.5\leq q\leq 1.
\end{equation}

Now to an upper bound. Since
$$
\frac{1+t^m}{2}\leq \frac{1+t^2}{2}\,,\quad m\geq 2,\quad 0\leq t\leq 1,
$$
then \eqref{eq6:3} implies
$$
\frac{w}{s}\leq \left(\frac{1+t}{2}\right)^N\frac{w}{\overline{w}}+\left(\frac{1+t^2}{2}\right)^N\sum_{m=2}^\infty\left(\frac{w}{\overline{w}}\right)=\frac{wq^N}{\overline{w}}+\frac{w^2(2q^2-2q+1)^N}{\overline{w}(\overline{w}-w)}\,,
$$
or
$$
\overline{w}-(w+sq^N)\overline{w}+ws\bigl(q^N-(2q^2-2q+1)^N\bigr)\leq 0.
$$
On solving this inequality, we find
\begin{equation}\label{eq6:11}
\overline{w}\leq \frac{w+sq^N}{2}+\sqrt{\left(\frac{w+sq^N}{2}\right)^2-ws\bigl(q^N-(2q^2-2q+1)^N\bigr)}.
\end{equation}
\begin{proposition}\label{prop:5:4}Consider the Eigen quasispecies model with the single peaked landscape $\bs w^\top=(w+s,w,\ldots,w)$ for $w,s>0$. If $0.5\leq q\leq 1$ then the mean fitness $\overline{w}(q)$ satisfies
\begin{equation}\label{eq6:7c}
   (w+s)q^N \leq \overline{w}(q)\leq \frac{w+sq^N}{2}+\sqrt{\left(\frac{w+sq^N}{2}\right)^2-ws\bigl(q^N-(2q^2-2q+1)^N\bigr)}.
\end{equation}
\end{proposition}

Together with Proposition \ref{prop:5:3}, we therefore have that the lower bound for $\overline{w}$ satisfies
\begin{equation}\label{eq6:9}
    \overline{w}(q)\geq \max\left\{w+s2^{-N},(w+s)q^N\right\},\quad 0.5\leq q\leq 1,
\end{equation}
and a similar formula is valid for the upper bound. However, it can be shown that the upper bound in \eqref{eq6:7c} is better than the upper bound in \eqref{eq6:7}.

Propositions \ref{prop:5:2}, \ref{prop:5:3}, and \ref{prop:5:4} together provide the proof for Proposition \ref{prop:main:1}.

The point at which the two lower bounds \eqref{eq6:9} intersect is given by
\begin{equation}\label{eq6:10}
    q^{\ast}=\sqrt[N]{1-\frac{s(1-2^{-N})}{w+s}}\,,
\end{equation}
which provides an estimate for the error threshold in the case of the single peaked landscape (see Section \ref{sec:3:2} for the discussion). Note that this formula can be written as
\begin{equation}\label{eq6:10a}
    q^\ast=\sqrt[N]{\frac{\overline{w}(0.5)}{\overline{w}(1)}}\,,
\end{equation}
which serves as a starting point for the conjecture discussed in Section \ref{sec:3:3}.

\subsection{A general remark on $\overline{w}(q)$ for an arbitrary fitness landscape}
Finally, we relate the specific results from Section \ref{sec:5} with an arbitrary fitness landscape $\bs w$.

Consider again the eigenvalue problem
\begin{equation*}
\bs{QWp}=\overline{w}\bs p,\tag{\ref{eq3:1}}
\end{equation*}
and matrix
$$
\sqrt{\bs W}=\diag(\sqrt{w_0},\ldots,\sqrt{w_{l-1}}).
$$
Rewrite \eqref{eq3:1} in the form
$$
\sqrt{\bs W}\bs Q\sqrt{\bs W}\sqrt{\bs W}\bs p=\overline{w}\sqrt{\bs W}\bs p,
$$
and denote $\bs s=\sqrt{\bs W}\bs p$, and $\bs C=\sqrt{\bs W}\bs Q\sqrt{\bs W}$. Vector $\bs s\in\R^l$ can be normalized as
$$
\bs s\cdot \bs s=1.
$$
We also have $\bs s\geq 0,\,\bs C\geq 0$ and $\bs C$ is symmetric. Then instead of \eqref{eq3:1} we have
$$
\bs{Cs}=\overline{w}\bs s,
$$
and hence
$$
\overline{w}=\bs s\cdot\bs {Cs}.
$$
Since $\bs s$ in nonnegative, the last expression can be rewritten as
$$
\overline{w}=\max\{\bs z\cdot \bs{Cz}\mid \bs z\geq 0,\,\bs z\cdot \bs z=1\}.
$$
Together with $\bs W$ consider another fitness landscape $\bs{\tilde{W}}$ such that $\bs{\tilde{W}}\geq \bs W$. Then, since $\bs Q\geq 0$, we have
$$
\bs{\tilde{C}}=\sqrt{\bs{\tilde{W}}}\bs Q\sqrt{\bs{\tilde{W}}}\geq \sqrt{\bs{{W}}}\bs Q\sqrt{\bs{{W}}}=\bs C\geq 0.
$$
This means that for any vector $\bs z\in \R^l$ such that $\bs z\geq 0$, $\bs z\cdot\bs z=1$
$$
\bs z\cdot \bs{\tilde{C}}\bs z\geq \bs z\cdot \bs{{C}}\bs z,
$$
or in terms of the leading eigenvalues $\overline{w}$:
$$
\tilde{\overline{w}}=\max\{\bs z\cdot \bs{\tilde{C}z}\mid \bs z\geq 0,\,\bs z\cdot \bs z=1\}\geq \max\{\bs z\cdot \bs{Cz}\mid \bs z\geq 0,\,\bs z\cdot \bs z=1\}= \overline{w}
$$
for any $q\in[0,1]$. This proves Proposition \ref{prop:3:3}.

In particular, it is true if $w_0=w+s$ is the maximal fitness of $\bs w$, $w=\max_{i\geq 1}\{w_i\}$, and $\bs{\tilde{w}}=\diag(w+s,w,\ldots,w)$.
\appendix
\section{Appendix}
\subsection{Explicit solution of the Eigen model}\label{append:A}
The Kronecker product, considered in Section \ref{propQ}, allows explicit solution for the multiplicative fitness landscape, which was treated in \cite{baake1997ising,higgs1994error,rumschitzki1987spectral}, and which we present here for future references.

Recall that for $\bs Q_N$ we have a natural recursive procedure
$$
\bs Q_N=\bs Q_1\otimes \bs Q_{N-1}, \quad N=1,\ldots,
$$
and
$$
\bs Q_1=
\begin{bmatrix}
  q & 1-q \\
  1-q & q \\
\end{bmatrix}.
$$
This natural recursive procedure allows to write down the eigenvalues and eigenvectors of $\bs Q_N$. If we were able to represent matrix $\bs{QW}$ in a similar form, then it would mean that we can find the leading eigenvalue (mean fitness) and the corresponding eigenvector (quasispecies). It turns out that if the fitness landscape is multiplicative, it is always possible to do. To wit, consider sequence type $i$:
$$
[a_0,a_1,\ldots,a_{N-1}],\quad a_{k}=\{0,1\}.
$$
and let $0<r_k<1$ be the contribution of the $k$-th site, $k=1,\ldots,N$ if there is ``1'' at the $k$-th site. Now define the fitness of sequence $i$ as
$$
w_i=\prod_{k\colon a_k=1}r_{k+1}, \quad i=1,\ldots,l-1, w_0=1,
$$
so that $r_k$ is a selective disadvantage to have the ``wrong'' letter at the $k$-th position.

Define matrices
$$
\bs W_k=\begin{bmatrix}
          1 & 0 \\
          0 & r_{k} \\
        \end{bmatrix},\quad k=1,\ldots,N.
$$
Now we have that
$$
(\bs Q\bs W)_k=(\bs Q_1\bs W_k)\otimes (\bs{QW})_{k-1},\quad k=1,\ldots,N,
$$
where $(\bs{QW})_k$ means the system matrix at the $k$-th step, and $(\bs Q\bs W)_1=\bs Q_1\bs W_1$. Now, since we know the maximal eigenvalue and the corresponding eigenvector for each $\bs Q_1\bs W_k$:
$$
\lambda_k=\frac 12 \left(q(1+r_k)+\sqrt{q^2(1+r_k)-4r_k(2q-1)}\right),
$$
$$
\bs v_k^\top=\left(1,C_k\right),\quad C_k=\frac{q(1-r)-\sqrt{q^2(1+r_k)-4r_k(2q-1)}}{2(q-1)}\,,
$$
the mean fitness is
$$
\overline{w}=\prod_{k=1}^N\lambda_k,
$$
and the quasispecies distribution (without the normalization) is
$$
\bs p=\bs v_N\otimes \bs v_{N-1} \otimes\ldots \otimes \bs v_1.
$$
This particular representation can be generalized to the cases when there are more than 2 possible letters per site, or when the mutation probabilities are site-dependent, to give a false sense of generality, but will work only for the multiplicative fitness landscape, i.e., biologically speaking, only for the case of no epistasis, which is not the most realistic case.
\subsection{Another approach to look for $\hat{q}$}\label{append:B}
Recall that we proved that there always exists the absolute minimum of $\overline{w}(q)$, which is achieved at $\hat{q}$ such that $\overline{w}'(\hat{q})=0$. In this appendix we present a procedure, which can be used to calculate $\hat{q}$ for small values of $N$.

Consider the characteristic equation
\begin{equation}\label{eq:B:1}
    \det(\bs{QW}-\lambda \bs I)=0,
\end{equation}
or, in explicit form
\begin{equation}\label{eq:B:2}
    c(\lambda)=\lambda^{2^N}+s_1(q,\bs w)\lambda^{l-1}+\ldots+s_{2^N}(q,\bs w)=0.
\end{equation}
Here $s_k(q,\bs w)$ are polynomials of $q$ depending on $\bs w$. We differentiate \eqref{eq:B:2} with respect to $q$ and obtain
\begin{equation}\label{eq:B:3}
    \left(2^N\lambda ^{l-1}+s_1(q,\bs w)(l-1)\lambda^{2^N-2}+\ldots\right)\lambda'+s_1'(q,\bs w)\lambda^{l-1}+\ldots+s'_{2^N}(q,\bs w)=0.
\end{equation}

If $q=\hat{q}$ then $\lambda'(q)=0$ and $\lambda=\hat{\lambda}$, and \eqref{eq:B:3} turns into
$$
c_1(\hat{\lambda})=s_1'(\hat q,\bs w)\hat{\lambda}^{l-1}+\ldots+s'_{2^N}(\hat q,\bs w)=0.
$$
Since $c(\lambda)$ and $c_1(\lambda)$ have a common root $\hat{\lambda}$ at $q=\hat{q}$, then their resultant
$$
R(q)=\text{Res}\,\bigl(c(\lambda),c_1(\lambda)\bigr)
$$
satisfies
$$
R(\hat{q})=0.
$$
The last expression can be used to determine $\hat{q}$ for modest values of $N$.
\begin{example}Let $N=2,l=2^N=4$ and
$$
\bs w=(1,4,9,16)^\top.
$$
We have that (Table \ref{tab:main:1})
$$
\overline{w}(0)=6,\quad \overline{w}(0.5)=7.5,\quad \overline{w}(1)=16,\quad \overline{w}'(0)=-2\overline{w}(0)=-12,\quad \overline{w}'(1)=2\overline{w}(1)=32.
$$
Direct calculations show that resultant $R(q)$ is
$$
R(q)=(2q-1)^7P_{10}(q),
$$
where $P_{10}(q)$ is some polynomial degree 10, that has 4 pairs of complex conjugate roots, one negative root, and one positive root $\hat{q}\approx 0.09317$, which is the sought minimum.
\end{example}

\subsection{Another approach to calculate $\overline{w}''(q)$}\label{append:C} In Section \ref{sec:4:5} we presented a formula for the second derivative:
$$
\bs{w}''=2 \bs z'\cdot(\overline{w}\bs I-\bs L)\bs z'+\bs z\cdot \bs{L}''\bs z,
$$
which uses for calculation the eigenvector $\bs z$ and its derivative. Here we show how we can find $\overline{w}''(q)$ through all the eigenvalues of $\bs{QW}$, or, equivalently, of $\bs{DF}$. We know that all the eigenvalues of $\bs{DF}$ are real, $\lambda_i\in\R,\,i=0,\ldots,l-1$. Assume that they are in decreasing order, such that $\overline{w}=\lambda_0$. Additionally assume that they all are different, at least for fixed $q=q_0$.

Consider again
$$
\bs{DFx}=\lambda\bs x,\quad \bs{FDy}=\lambda\bs y.
$$
Let $\bs x_0,\ldots, \bs x_{l-1}$ be the linearly independent eigenvectors of $\bs{DF}$ (they are real and always exist) and $\bs y_0,\ldots,\bs y_{l-1}$ be the corresponding eigenvectors of $\bs{FD}$ that satisfy, for fixed $q=q_0$,
$$
\bs x_i\cdot \bs y_j=\delta_{ij},
$$
where $\delta_{ij}$ is the Kronecker delta. Without loss of generality, we can assume that
\begin{equation}\label{eq:C:1}
    \bs y_0(q_0)\cdot \bs x_0(q)=1
\end{equation}
for $q$ close to $q_0$. Differentiating \eqref{eq:C:1} and plugging $q=q_0$, we find
$$
\bs y_0\cdot \bs x_0'=0,\quad \bs y_0\cdot \bs x_0''=0
$$
at the point $q=q_0$.
\begin{proposition}Given the conditions specified above and for $q\neq 0.5$, one has
$$
\left(q-\frac 12\right)^2\lambda_0''=\lambda_0\left(\bs y_0\cdot (\Delta^2-\Delta)\bs x_0+2\sum_{j=1}^{l-1}\frac{\lambda_j}{\lambda_0-\lambda_j}\bs (\bs y_j\cdot \Delta\bs x_0)(\bs y_0\cdot\Delta\bs x_j)\right),
$$
where $\Delta=\diag(0,1,\ldots,H_i,\ldots,N)$.
\end{proposition}
\begin{proof}
Since $\{\bs x_0,\ldots,\bs x_{l-1}\}$ is a basis in $\R^l$, we always can write
$$
\bs x_0'=C_0\bs x_0+\ldots+C_{l-1}\bs x_{l-1}.
$$
for some constants $C_j$. By multiplying this by $\bs y_0$ on the left and taking into account the normalizations, we find $C_0=0$, i.e.,
$$
\bs x_0'=C_1\bs x_1+\ldots+C_{l-1}\bs x_{l-1}.
$$
Differentiate $\bs {DFx}_0=\lambda_0\bs x_0$ with respect to $q$ implies
$$
\bs D'\bs F\bs x_0+\bs {DFx}_0'=\lambda_0'\bs x_0+\lambda_0\bs x_0'.
$$
By multiplying this equality by $\bs y_j$ we find
$$
\bs y_j\cdot \bs D'\bs F\bs x_0=(\lambda_0-\lambda_j)\bs y_j\cdot \bs x_0'=(\lambda_0-\lambda_j)C_j,
$$
and hence
$$
C_j=\frac{\bs y_j\cdot \bs D'\bs F\bs x_0}{\lambda_0-\lambda_j}\,,\quad j=1,\ldots, l-1.
$$
Now find the second derivative of $\bs{DFx}_0=\lambda_0\bs x_0$:
$$
\bs D''\bs{Fx}_0+2\bs D'\bs F\bs x_0'+\bs{DF}\bs x_0''=\lambda_0''\bs x_0+2\lambda'_0\bs x_0'+\lambda_0\bs x_0'',
$$
from where, after multiplication by $\bs y_j$ and some simplifications,
$$
\lambda_0''=\bs y_0\cdot \bs D''\bs F\bs x_0+2\bs y_0\cdot\bs D'\bs F\bs x_0'.
$$
Let $q\neq 0.5$, then $\bs F\bs x_0=\lambda_0\bs D^{-1}\bs x_0$, and using the expressions from Section \ref{sec:4:5} for $\bs D'$ and $\bs D''$, we obtain
$$
\bs y_0\cdot\bs D''\bs F\bs x_0=\lambda_0\bs y_0\cdot \bs D''\bs D^{-1}\bs x_0=\frac{4\lambda_0}{(2q-1)^2}\bs y_0\cdot (\Delta^2-\Delta)\bs x_0.
$$
Moreover, the representation of $\bs x'_0$ and the expressions for $C_j$ yield
$$
2\bs y_0\bs D'\bs F\bs x_0'=2\sum_{j=1}^{l-1}C_j\bs y_0\cdot \bs D'\bs F\bs x_j=2\sum_{j=1}^{l-1}\frac{(\bs y_j\cdot \bs D'\bs F\bs x_0)(\bs y_0\cdot\bs D'\bs F\bs x_j)}{\lambda_0-\lambda_s}\,.
$$
We have $\bs F\bs x_j=\lambda_j\bs D^{-1}\bs x_j$ and $\bs D'=\frac{2}{2q-1}\bs D\Delta$, and therefore
$$
(\bs y_j\cdot \bs D'\bs F\bs x_0)(\bs y_0\cdot\bs D'\bs F\bs x_j)=\frac{4\lambda_0\lambda_j(\bs y_j\cdot\Delta\bs x_0)(\bs y_0\cdot\Delta\bs x_j)}{(2q-1)^2}\,.
$$
Now if we plug everything into the expression for $\lambda''_0$, after some rearrangement we find the announced formula.
\end{proof}

\paragraph{Acknowledgements:} This research is supported in part by the Russian Foundation for Basic Research (RFBR) grant \#10-01-00374 and joint
grant between RFBR and Taiwan National Council \#12-01-92004HHC-a. ASN's research is supported in part by ND EPSCoR and NSF grant \#EPS-0814442.


\end{document}